\documentclass[twocolumn, pra, longbibliography, superscriptaddress]{revtex4-2}
\usepackage{hyperref}
\usepackage{amsmath,mathrsfs,amsthm}
\usepackage{graphicx}%
\usepackage{amsfonts}%
\usepackage{amssymb}
\usepackage{romannum}
\usepackage[braket,qm]{qcircuit}
\newtheorem{theorem}{Theorem}
\newtheorem{lemma}{Lemma}
\newtheorem{definition}{Definition}

\begin{document}
\title{Temporal information processing on noisy quantum computers}

\author{Jiayin \surname{Chen}}
\affiliation{School of Electrical Engineering and Telecommunications, The University of New South Wales (UNSW), Sydney NSW 2052, Australia.}
\affiliation{Quantum Computing Center, Keio University, Hiyoshi 3-14-1, Kohoku, Yokohama 223-8522, Japan.}
\author{Hendra I. \surname{Nurdin}}
\email{h.nurdin@unsw.edu.au}
\affiliation{School of Electrical Engineering and Telecommunications, The University of New South Wales (UNSW), Sydney NSW 2052, Australia.}
\author{Naoki \surname{Yamamoto}}
\affiliation{Quantum Computing Center, Keio University, Hiyoshi 3-14-1, Kohoku, Yokohama 223-8522, Japan.}
\affiliation{Department of Applied Physics and Physico-Informatics, Keio University, Hiyoshi 3-14-1, Kohoku, Yokohama 223-8522, Japan.}

\begin{abstract} 
The combination of machine learning and quantum computing has emerged as a promising approach for addressing previously untenable problems. Reservoir computing is an efficient learning paradigm that utilizes nonlinear dynamical systems for temporal information processing, i.e., processing of input sequences to produce output sequences. Here we propose quantum reservoir computing that harnesses complex dissipative quantum dynamics. Our class of quantum reservoirs is universal, in that any nonlinear fading memory map can be approximated arbitrarily closely and uniformly over all inputs by a quantum reservoir from this class. We describe a subclass of the universal class that is readily implementable using quantum gates native to current noisy gate-model quantum computers. Proof-of-principle experiments on remotely accessed cloud-based superconducting quantum computers demonstrate that small and noisy quantum reservoirs can tackle high-order nonlinear temporal tasks. Our theoretical and experimental results pave the path for attractive temporal processing applications of near-term gate-model quantum computers of increasing fidelity but without quantum error correction, signifying the potential of these devices for wider applications including neural modeling, speech recognition and natural language processing, going beyond static classification and regression tasks.
\end{abstract}

\maketitle
\pagenumbering{arabic}

\section{Introduction} 
The ingenious use of quantum effects has led to a significant number of quantum machine learning algorithms that offer computational speed-ups \cite{biamonte2017quantum,ciliberto2018quantum}. While awaiting the demonstration of these quantum algorithms on full-fledge quantum computers equipped with quantum error correction, quantum computing has transitioned from theoretical ideas to the noisy intermediate-scale quantum (NISQ) technology era \cite{Preskill18}. Hybrid quantum-classical algorithms using short-depth circuits are particularly suitable for implementation on NISQ devices. Many notable experimental demonstrations of NISQ devices employ hybrid algorithms for data classification \cite{havlivcek2019supervised} and quantum chemistry \cite{kandala2017hardware}. An on-going quest is to find interesting applications on quantum computers with increasingly lower noise profile but not reaching a low enough threshold to enable continuous quantum error correction. 

Here we propose a hybrid quantum-classical algorithm that \textit{utilizes dissipative quantum dynamics} as reservoir computers (RC) for temporal information processing on gate-model NISQ quantum computers. The goal for temporal information processing tasks, such as speech processing and natural language processing \cite{MH19,HR19}, is to learn the relationship between input sequences and output sequences. The RC framework uses an arbitrary but fixed dynamical system (in this case systems with dynamics described by state-space difference equation), the ``reservoir'', to map sequential inputs into its high-dimensional state-space. Only a simple linear regression algorithm is required to optimize the parameters of a readout function to approximate target outputs. The use of a simple linear readout has connections to the biological concept of mixed selectivity, as demonstrated in monkeys \cite{RBWWDMF13}. The attractiveness of the RC scheme is that naturally occurring dynamical systems (with some desired properties) in physics and engineering can be exploited for temporal information processing without precise tuning of its parameters, circumventing the expensive training cost in alternative schemes such as recurrent neural networks with tunable internal weights \cite{pascanu2013difficulty}. The ease of RC implementation has brought forward many successful hardware implementations of classical (i.e., non-quantum) RC schemes \cite{du2017reservoir,vandoorne2014experimental,vinckier2015high}. A spintronic RC achieved state-of-the-art performance on a spoken digit recognition task \cite{torrejon2017neuromorphic} and a photonic RC demonstrated high-speed speech classification with a low error \cite{larger2017high}. For theoretical developments, \cite{gonon2020approximation} derives an approximation error upper bound for certain classical RCs on learning a class of input-output maps (not necessarily fading memory maps considered here). Information processing capacity of various RC schemes has been investigated \cite{dambre2012information,gonon2020memory}. See \cite{tanaka2019recent,nakajima2020physical} for further interests and developments of RCs.

In this work, we employ dissipative quantum systems as quantum reservoirs (QRs) to approximate nonlinear input-output maps with fading memory. A map has fading memory if its outputs depend increasingly less on inputs from earlier times. These maps are important in a broad class of real-world problems including spoken digit recognition \cite{torrejon2017neuromorphic} and neural modeling \cite{MNM02}. The use of quantum systems as QRs was initially proposed in \cite{fujii2017harnessing,nakajima2019boosting} to harness disordered-ensemble quantum dynamics for temporal information processing. This QR class is suitable for ensemble quantum systems and a static (non-temporal) version of \cite{fujii2017harnessing} was demonstrated in NMR to approximate static maps \cite{negoro2018machine}. However, it remained an open problem to show this QR class has the properties required for reservoir computing. Chen and Nurdin \cite{chen2019learning} addressed this problem by demonstrating that a variation of the scheme proposed in \cite{fujii2017harnessing,nakajima2019boosting} is universal for nonlinear fading memory maps, meaning that given any target nonlinear fading memory map, there exists a member in the universal QR class whose outputs approximate the target map's outputs arbitrarily closely and uniformly over the input sequences. This is a quantum analogue of the universal function approximation property feed-forward neural networks enjoy \cite{cybenko1989approximation, hornik1989multilayer}, but for nonlinear fading memory mappings from input sequences to output sequences. The notion of universality we adopt here was previously established for classical RC schemes \cite{grigoryeva2018universal,GO18,MNM02} and the Volterra series \cite{BC85}. In particular, \cite{GO18} proves this universality property for a form of recurrent neural networks called echo-state networks \cite{JH04}. However, realizing these previous QR proposals in the quantum gate-model remains challenging due to the large number of quantum gates required to implement the dynamics via Trotterization.

The contribution of this work is twofold. Firstly, we propose a new class of QRs endowed with the fading memory and universality properties that is not implemented by Ising Hamiltonians, circumventing the need for Trotterization required in previous proposals. Secondly, we propose a realization of a subclass of the universal QR class on NISQ devices and present proof-of-principle experiments on remotely accessed IBM superconducting quantum processors \cite{IBMQ}, i.e., NISQ devices not yet equipped with quantum error correction. The QR dynamics in this subclass can be implemented using arbitrary but fixed quantum circuits, as long as they generate non-trivial dynamics. This could be, for instance, quantum circuits that are classically intractable to simulate. The quantum circuits can be of short lengths and can be implemented using parametrized single-qubit and multi-qubit quantum gates native to the quantum hardware, without the need for precise tuning of their gate parameters. Our proof-of-principle experiments show that QRs with a small number of qubits operating in a noisy environment can tackle complex nonlinear temporal tasks, even under current hardware limitations and in the absence of readout and process error mitigation techniques. This work serves as the first theoretical and experimental realization of applying near-term gate-model quantum computers to nonlinear temporal information processing tasks, opening an avenue for time series modeling and signal processing applications of these devices.

The rest of this paper is organized as follows. In Sec.~\ref{sec:temporal-task} we introduce fading memory maps and describe two temporal information processing tasks for these maps. Sec.~\ref{sec:RC} introduces the RC framework and explains conditions for which a RC defines a fading memory map. Sec.~\ref{sec:universal-QRC} presents our QR proposal and the universality result. We then propose a subclass  of the universal class suitable for implementation on current noisy gate-model quantum computers. We conclude the section by discussing invariance properties of the universal class under certain hardware imperfections. Sec.~\ref{sec:realization} details two hardware realizations of the aforementioned subclass of the universal one and presents more efficient versions of both schemes that could enable QR's potential for more scalable temporal processing on gate-model quantum devices. Sec.~\ref{sec:exp} details our proof-of-principle experiments performed on cloud-based IBM superconducting quantum devices. We provide concluding remarks in Sec.~\ref{sec:conclusion}. Detailed mathematical derivations and experimental settings are provided in the Appendix.

\section{Temporal information processing}
\label{sec:temporal-task}
We consider a input-output (I/O) map $M$ that maps infinite input sequences $u= \{\ldots, u_{-1}, u_{0}, u_{1}, \ldots \}$ to infinite output sequences $y = \{ \ldots, y_{-1}, y_{0}, y_{1}, \ldots \}$, where $u_l , y_l \in \mathbb{R}$ for $l\in \mathbb{Z}$ and $y_l = M(u)_l$ is the output at time $l$. We write $u\rvert_{L:L'} = \{ u_L, \ldots, u_{L'}\}$ and $y\rvert_{L:L'} = \{y_L, \ldots, y_{L'}\}$ to denote the inputs and outputs during time $l = L, \ldots , L'$. In practice, such I/O maps can be realized by convergent dynamical systems, that is, systems that forget their initial condition (see Appendix~\ref{app-subsec:cv} for details). If such a dynamical system with state $\textbf{x}_l$ is initialized at time $l_0$ at the state $\textbf{x}_{l_0}$ and given an input sequence $\{u_{l_0}, u_{l_0+1},\ldots\}$ and the system outputs the sequence $\{y_{l_0},y_{l_0+1},\ldots\}$, then it realizes an I/O map $M$ for any initial condition $\textbf{x}_{l_0}$ as $l_0 \rightarrow -\infty$.

Two challenging temporal information processing problems are posed to learn the I/O relationship given by $M$ based on the I/O pair $u, y$. The first is the multi-step ahead prediction problem, in which we are given inputs $u\rvert_{1:L}$ and the corresponding outputs $y\rvert_{1:L}$. The first $L_T < L$ input-output data pair ($u\rvert_{1:L_T}, y\rvert_{1:L_T}$) is the train data. In the sequel, we use the input-output train data during $l=5, \ldots, L_T$. The reason for this is to remove the transient response in the data, see Sec.~\ref{subsec:exp-im} for a discussion. The goal is to use the train data to optimize the parameters $\textbf{w}$ of another I/O map $\overline{M}_\textbf{w}$, so that the outputs $\overline{y}\rvert_{L_T+1:L}=\{\overline{y}_{L_T+1}, \ldots \overline{y}_L\}$, where $\overline{y}_l = \overline{M}_\textbf{w}(u)_l$, approximate the target outputs $y\rvert_{L_T + 1 : L}$. The second problem is the map emulation problem, that is to optimize $\textbf{w}$ of $\overline{M}_\textbf{w}$ to emulate $M$ using $k=1, \ldots, K$ different I/O train data pairs $(u^k\rvert_{1:L'}, y^k\rvert_{1:L'})$, so that the total number of train data is $KL'$ (we will again use train data during $l=5, \ldots, l=L'$ in Sec.~~\ref{subsec:exp-im}). When given a previously unseen input $u^{K+1}\rvert_{1:L'}$, the task is for $\overline{y}^{K+1}\rvert_{1:L'}$ to approximate $y^{K+1}\rvert_{1:L'}$.

If an I/O map $M$ has fading memory, then its output at time $l'$ becomes increasingly less dependent on input samples $u_l$ from much earlier times $l \ll l'$; see Appendix~\ref{app-subsec:fmp}. In this work, we approximate nonlinear fading memory I/O maps using RCs implemented by quantum dynamical systems. We will introduce conditions for which a reservoir dynamical system defines a fading memory I/O map in the next section.

\section{Reservoir computing}
\label{sec:RC}
To approximate fading memory maps, RC exploits nonlinear dynamical systems to project the input $u_l$ into a reservoir state $\textbf{x}_l$ at time $l$. A RC is governed by a dynamics $f$ with state evolution $\textbf{x}_{l} = f(\textbf{x}_{l-1}, u_l)$. The dynamics of the reservoir can be arbitrary but fixed as long as it satisfies some required properties, and never requires training. We require the RC to satisfy the echo-state property \cite{JH04} or the convergence property \cite{PWN05}, so that the RC asymptotically forgets its initial condition. The tunable parameters $\textbf{w}$ appear in a readout function $h_{\textbf{w}}$, which combines the elements of $\textbf{x}_l$ into an output $\overline{y}_l = h_{\textbf{w}}(\textbf{x}_l)$. For a sufficiently long input sequence $\{u_{l_0}, u_{l_0+1}, \ldots, u_{0}\}$, the effect of the RC's initial condition can be washed-out. As discussed in Sec.~\ref{sec:temporal-task}, as $l_0 \rightarrow -\infty$, the combination of a convergent RC dynamics $f$ and the readout function $h_\textbf{w}$ produces an I/O map $\overline{M}_{(f, h_\textbf{w})}$. After the washout, the readout parameters $\textbf{w}$ can be optimized using linear regression to minimize an empirical mean squared-error between $y_{1:L_T}$ and $\overline{y}_{1:L_T}$. As in previous works \cite{BC85,MNM02,grigoryeva2018universal,GO18}, we consider $\overline{M}_{(f, h_\textbf{w})}$ that has the fading memory property.

Echo-state networks, one of the pioneering classical RC schemes, have been numerically demonstrated to achieve state-of-the-art performance in chaotic system modeling \cite{JH04}. Subsequent hardware realizations of RC proposals exploit classical dynamical systems for real-time temporal processing tasks that demand less energy or computational memory \cite{du2017reservoir,vandoorne2014experimental,vinckier2015high,torrejon2017neuromorphic,larger2017high}. These experiments also suggest empirically that for certain tasks, such as spoken digit recognition, the reservoir state dimension plays a role in the RC's task performance.

\section{Universal quantum reservoir computers}
\label{sec:universal-QRC}
We propose to use a QR, with a view towards possibly taking advantage of fast quantum dynamics and its exponentially large state space. A QR consists of $N$ non-interacting subsystems, each subsystem $k$ has $n_k$ number of qubits so that the QR has $n=\sum_{k=1}^{N}n_k$ qubits. The QR density operator $\rho_l$ at time $l$ evolves according to
\begin{equation} \label{eq:dynamics}
\rho_l = T(u_l) \rho_{l-1} = \bigotimes_{k=1}^{N} T^{(k)}(u_l) \rho^{(k)}_{l-1},
\end{equation}
and the $k$-th subsystem density operator $\rho^{(k)}_l$ undergoes the evolution
\begin{equation} \label{eq:dynamics-subsystem}
\begin{split}
& T^{(k)}(u_{l})\rho^{(k)}_{l-1} \\
& = (1-\epsilon_k)\left(u_{l} T^{(k)}_{0} + (1-u_{l}) T^{(k)}_1\right)\rho^{(k)}_{l-1} + \epsilon_k \sigma_k,
\end{split}
\end{equation}
for input $0 \leq u_{l} \leq 1$. Here, $0 < \epsilon_k \leq 1$, $\sigma_k$ is an arbitrary but fixed density operator, and $T^{(k)}_{0}$ and $T^{(k)}_{1}$ are two arbitrary but fixed completely positive trace-preserving (CPTP) maps. Examples of such maps include some naturally occurring noisy quantum channels, such as dephasing or amplitude damping channels; see \cite{NC10}. No precise tuning or engineering of the CPTP maps $T^{(k)}_0, T^{(k)}_1$ is required for the QR scheme and it should not generate trivial dynamics (i.e., we should not choose $T^{(k)}_0 = T^{(k)}_1$). They could potentially be classically intractable to simulate CPTP maps. The QR dynamics Eq.~\eqref{eq:dynamics}--\eqref{eq:dynamics-subsystem} is convergent, meaning that it will asymptotically forget its initial condition; see Appendix~\ref{app-subsec:cv} for the proof. Given inputs $\{ u_{l_0}, u_{l_0 + 1}, \ldots,u_0\}$ and $l_0 \rightarrow -\infty$, the convergence property ensures the QR state $\rho_0$ evolves according to Eqs.~\eqref{eq:dynamics}--\eqref{eq:dynamics-subsystem} is determined by $\{ u_{l_0}, u_{l_0 + 1}, \ldots,u_0\}$ and $T^{(k)}_0, T^{(k)}_1$, but not by its initial state $\rho_{l_0}$.

We obtain partial information about $\rho_l$ by measuring each qubit in the Pauli $Z$ basis to obtain $\langle Z^{(i)} \rangle_l = {\rm Tr}(\rho_l Z^{(i)})$ for $i=1,\ldots,n$, where $Z^{(i)}$ acts on qubit $i$. We associate the readout function Eq.~\eqref{eq:readout} to the QR dynamics Eq.~\eqref{eq:dynamics}. The readout function Eq.~\eqref{eq:readout} is a multivariate polynomial of degree $R$ in the variables $\langle Z^{(i_j)} \rangle_l$. A simple linear form ($R=1$) is employed in our proof-of-principle experiments in Sec.~\ref{sec:exp}. The tunable readout parameters $\textbf{w} = \{ \textbf{w}_{i_1, \ldots, i_n}^{r_{i_1}, \ldots, r_{i_n}}, \textbf{w}_c \}$ can be optimized via linear regression. Eqs.~\eqref{eq:dynamics} and \eqref{eq:readout} define a QR implementing an I/O map $\overline{M}_{(T,h_\textbf{w})}$ that depends on the QR dynamics $T$ and the readout function $h_\textbf{w}$. We show in Appendix~\ref{app-subsec:fmp} that $\overline{M}_{(T, h_\textbf{w})}$ has the fading memory property. Now consider the class $\mathcal{M}$ of I/O maps $\overline{M}_{(T, h_\textbf{w})}$ arising from differing numbers of subsystems $N$, numbers of qubits $n$, QR dynamics $T(u_l)$, readout parameters $\textbf{w}$ and degree $R$ of $h_\textbf{w}$. Our main result shows that the class $\mathcal{M}$ is universal for approximating nonlinear fading memory maps.

\begin{theorem}[Universality]
Let $K([0, 1])$ be the set of input sequences $\{u_l\}$ with $0 \leq u_l \leq 1$ for $l \in \mathbb{Z}$. For any nonlinear fading memory map $M$ and any $\delta > 0$, there exists $\overline{M}_{(T, h_\textbf{w})} \in \mathcal{M}$ implemented by some QR such that for all $u \in K([0, 1])$, $\sup_{l \in \mathbb{Z}} \left| M(u)_l - \overline{M}_{(T, h_\textbf{w})}(u)_l \right| < \delta$.
\end{theorem}
We remark that universality is a property of the QR class $\mathcal{M}$ and not of an individual member of $\mathcal{M}$. The universality proof employs the Stone-Weierstrass Theorem \cite[Theorem~7.3.1]{Dieudonne13}, see Appendix~\ref{app-subsec:universal} for the proof. Besides the universality property, our proposed universal QR class exhibits invariance properties under certain hardware imperfections; see Sec.~\ref{subsec:invariance} below.

\begin{widetext}
\begin{equation}
\vspace*{-3.5em}
\label{eq:readout}
\overline{y}_{l}  = h_\textbf{w}(\rho_{l}) = \sum_{d=1}^{R} \sum_{i_{1} = 1}^{n}  \cdots \sum_{i_{n} = i_{n-1}+1}^{n}  \sum_{r_{i_1} + \cdots + r_{i_n} = d} \textbf{w}_{i_1, \ldots, i_{n}}^{r_{i_1}, \ldots, r_{i_n}} \langle Z^{(i_{1})} \rangle_{l}^{r_{i_1}} \cdots \langle Z^{(i_n)} \rangle_{l}^{r_{i_n}} + \textbf{w}_c,
\end{equation}
\end{widetext}

\begin{figure}
\includegraphics[scale=1, trim={2em 6em 2em 6em}, clip]{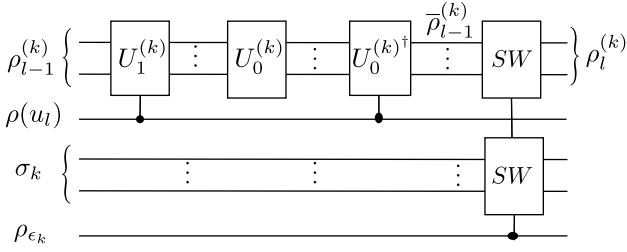}
\caption{Quantum circuit interpretation of the QR universal subclass described in Sec.~\ref{subsec:subclass}. Here $\rho^{(k)}_{l-1}$ and $\sigma_k$ are two quantum registers (i.e., groups of qubits) whereas $\rho(u_l)$ and $\rho_{\epsilon_k}$ are two single-qubit states. The unitaries $U^{(k)}_1, U^{{(k)}^\dagger}_0$ act on $\rho^{(k)}_{l-1}$, controlled by $\rho(u_l)$. The right-most operation ($SW$'s) swaps the states of $\overline{\rho}^{(k)}_{l-1}$ and $\sigma_k$, controlled by $\rho_{\epsilon_k}$.}
\label{fig:1} 
\vspace*{-0.5cm}
\end{figure}

\subsection{A subclass implementable on noisy gate-model quantum devices}
\label{subsec:subclass}
With a limited number of qubits and other current quantum hardware restrictions, not all QR dynamics of the form Eqs.~\eqref{eq:dynamics}--\eqref{eq:dynamics-subsystem} can be efficiently implemented. Here we describe a subclass of the universal QR class implementable on current gate-model quantum devices.

QRs in this subclass are governed by Eqs.~\eqref{eq:dynamics}--\eqref{eq:dynamics-subsystem} with unitary evolutions $T^{(k)}_{j}(\rho^{(k)}_{l-1}) = U^{(k)}_j \rho^{(k)}_{l-1} U^{(k)^\dagger}_j$ ($j=0, 1$), where the unitaries $U^{(k)}_0$ and $U^{(k)}_1$ are arbitrary but fixed. In practice, $U^{(k)}_j$ can be implemented by native quantum gates of the NISQ devices, possibly composed of single-qubit and multi-qubit gates each parameterized by some gate parameter. These gate parameters can be chosen arbitrarily but fixed and should not generate trivial dynamics (e.g., we should not have $U^{(k)}_0= U^{(k)}_{1}$), thus precise tuning of these parameters is not required. In Sec.~\ref{subsec:qr-circuits}, we suggest some natural choices of $U^{(k)}_j$ tailored for the cloud-based IBM quantum devices \cite{IBMQ}. The QR dynamics in this subclass has a natural quantum circuit interpretation, see Fig.~\ref{fig:1}. The state $\rho(u_l)$ encodes the input $u_l$ as a classical mixture $\rho(u_l) = u_l |0 \rangle \langle 0| + (1-u_l) |1 \rangle \langle 1|$, meaning that we apply $U^{(k)}_0 \rho^{(k)}_{l-1} U^{(k)^\dagger}_0$ with probability $u_l$, and apply $U^{{(k)}^\dagger}_0 U^{(k)}_0 U^{(k)}_1 \rho^{(k)}_{l-1} U^{{(k)}^\dagger}_1 U^{{(k)}^\dagger}_0 U^{(k)}_0 = U^{(k)}_1 \rho^{(k)}_{l-1} U^{{(k)}^\dagger}_1$ with probability $1-u_l$. Let $\overline{\rho}^{(k)}_{l-1}$ denote the QR's $k$-th subsystem state after these operations. The state $\rho_{\epsilon_k}$ is a classical mixture $\rho_{\epsilon_k} = (1-\epsilon_k) | 0 \rangle \langle 0| + \epsilon_k |1 \rangle \langle 1|$ that encodes the rate $\epsilon_k$ at which the $k$-th subsystem forgets its initial conditions. That is, with probability $\epsilon_k$, the states $\overline{\rho}^{(k)}_{l-1}$ and $\sigma_k$ are exchanged, equivalent to resetting the state  $\overline{\rho}^{(k)}_{l-1}$ to the fixed density operator $\sigma_k$; otherwise the state $\overline{\rho}^{(k)}_{l-1}$ is unchanged with probability $1-\epsilon_k$. We again associate the readout function Eq.~\eqref{eq:readout} to this QR subclass.

\subsection{Invariance under stationary Markovian hardware noise and time-invariant readout error}
\label{subsec:invariance}
The QR dynamics Eq.~\eqref{eq:dynamics} is invariant under stationary Markovian noise. A stationary Markovian noise process acting on the $k$-th subsystem during some time interval $\tau(l-1) \leq t \leq \tau l$, where $l$ is the time step and $\tau>0$, can be modeled as a CPTP map $\mathcal{T}^{(k)} $ for all $l \geq 0$. The $k$-th subsystem's dynamics Eq.\eqref{eq:dynamics-subsystem} under this noise process is
\begin{equation*}
\begin{split}
& \rho^{(k)}_{l} = (1-\epsilon_k) \left(u_l \mathcal{T}^{(k)} \circ T^{(k)}_0 + (1-u_l) \mathcal{T}^{(k)} \circ T^{(k)}_1 \right) \rho^{(k)}_{l-1} \\
& \hspace*{3em} + \epsilon_k \mathcal{T}^{(k)}(\sigma_k),
\end{split}
\end{equation*}
where $\mathcal{T}^{(k)} \circ T^{(k)}_{j}$ is again some CPTP for $j=0, 1$ and $\mathcal{T}^{(k)}(\sigma_k) = \sigma'_k$ is again some fixed density operator. The resulting noisy dynamics again has the form Eq.~\eqref{eq:dynamics-subsystem} and the form of QR dynamics Eq.~\eqref{eq:dynamics} also remains unchanged. That is, the universal family $\mathcal{M}$ is \textit{invariant and remains universal} under stationary Markovian noise. For hardware implementation of the QR subclass described in Sec.~\ref{subsec:subclass}, if the hardware noise is stationary and Markovian, then it acts to replace $U^{(k)}_j \rho^{(k)}_{l-1} U^{{(k)}^\dagger}_j$ with another CPTP map $\overline{T}^{(k)}_j(\rho^{(k)}_{l-1})$. The resulting noisy QR dynamics is again of the form Eq.~\eqref{eq:dynamics}.

Stationary Markovian noise model is the noise model adopted in the IBM Qiskit simulator \cite{Aernoise,NoiseModel}. The Qiskit noisy simulation approximates the hardware noise  as a CPTP map being applied after the application of a unitary gate. The noise parameters are estimated during periodic calibrations on the hardware. Between two calibrations, the calibrated noise parameters remain unchanged and the noisy simulation approximates the hardware noise by a stationary Markovian noise model. However, during the experiments, the underlying hardware noise could potentially be time-varying. Considering these factors, the agreement between our experimental and Qiskit noisy simulation results (see Appendix~\ref{app-subsec:measurement-data} for the data) indicate the underlying hardware noise approximately preserves the QR dynamics of the form Eq.~\eqref{eq:dynamics} during the experiments. If the underlying noise is non-stationary but changes slowly, the QR output weights can be re-trained periodically using most recently gathered data. This remains challenging to be demonstrated on current cloud-accessed only NISQ devices but can be possible on future NISQ machines.

Furthermore, QR predicted outputs remain unchanged under time-invariant readout error whenever a linear readout function is used (i.e., $R=1$ in Eq.~\eqref{eq:readout}, which is often employed in practice and in our proof-of-principle experiments). This is because time-invariant readout error introduces a time-invariant linear transformation of the measurement data and if the output weights $\textbf{w}_{i_1, \ldots, i_n}^{r_{i_1}, \ldots, r_{i_n}}$ and $\textbf{w}_c$ are optimized via linear regression, the resulting QR predicted outputs $\overline{y}_l$ remain unchanged; see Appendix~\ref{app-sec:readout-error} for the derivation.

\section{Realization of a subclass on current quantum hardware}
\label{sec:realization}

We present two implementation schemes of the subclass described in Sec.~\ref{subsec:subclass} on current gate-model quantum computers, such as on the IBM superconducting quantum devices. The first scheme takes into account limitations of some current hardware, and the second scheme employs quantum non-demolition (QND) measurements to substantially reduce the number of circuit runs required. We further show that QR's convergence property leads to more efficient versions of both schemes. Here, we focus on $n$-qubit QRs with a single subsystem ($N=1$ in Eq.~\eqref{eq:dynamics}) and drop the subsystem index $k$ in Eq.~\eqref{eq:dynamics-subsystem}. The case with multiple subsystems ($N > 1$) is a straightforward extension. We may choose $\sigma = |\psi \rangle \langle \psi|$ with an easy to prepare pure state $|\psi\rangle$. In all schemes, we initialize the QR circuits in $|0\rangle^{\otimes n}$.

The first implementation follows from an earlier work \cite[Sec.~III]{JNY19} and is employed in our proof-of-principle experiments (see Sec.~\ref{sec:exp}). We consider NISQ devices that allow pure state preparation. Instead of realizing Fig.~\ref{fig:1} that requires mixed state preparation, we efficiently implement QRs through Monte Carlo sampling. We construct $N_m$ circuits, such that for each circuit and at each timestep $l$, we apply $U_0$ and $U_1$ with probabilities $(1-\epsilon)u_{l}$ and $(1-\epsilon)(1-u_l)$, respectively; otherwise the circuit is set in $|\psi \rangle$ with probability $\epsilon$. Therefore, for each $N_m$ circuits and each time $l$, implementing the input-dependent QR dynamics $T(u_l)$ in Eq.~\eqref{eq:dynamics} amounts to applying the gate sequence realizing $U_0$ or $U_1$, or resetting the circuit in $|\psi \rangle$. As $N_m$ is increased, the average of all measurements gives a more accurate estimate of the true expectation $\langle Z^{(i)} \rangle_l$. Furthermore, some current NISQ devices do not allow qubit reset, meaning that once a qubit is measured, it cannot be re-used for computation. To estimate $\langle Z^{(i)} \rangle_l$, we re-initialize $N_m$ circuits in $| 0 \rangle^{\otimes n}$ and re-apply $T(u_k)$ from time $k=1$ to time $k=l$, and only measuring $Z^{(i)}$ at the final time $l$. Each of the $N_m$ circuits is run for $S$ shots at each time $l$. To process a length-$L$ input sequence under the pure state and qubit re-set limitations requires $N_m S L$ circuit runs and $N_m S (1 + \cdots + L) = N_m S (L+1) L / 2$ applications of $T(u_l)$.

If qubit reset is available, a more efficient scheme using QND measurements \cite{braginsky1995quantum} can be realized, see Appendix~\ref{app-sec:qnd} for the details. We no longer need to re-run the $N_m$ circuits from time $1$ to estimate $\langle Z^{(i)} \rangle_l$. Instead we just run each of the $N_m$ circuits $S$ shots, meaning that for each circuit we perform a QND measurement of $Z^{(i)}$ at time $l$, continue running the circuit until the next measurement, and so forth. QND measurements ensure information encoded in $\rho_l$ is retained from one timestep to the next. This scheme requires $N_m S L$ applications of $T(u_l)$ but only $N_m S$ circuit runs as opposed to $N_m S L$ runs in the first scheme. We remark that a recent noisy quantum device is equipped with the qubit reset functionality \cite{pino2020demonstration}, and it will be interesting to implement this scheme in such a device in a future work.

The QR's convergence property (see Appendix~\ref{app-subsec:cv}) leads to more efficient versions of both schemes. Let $M \geq 1$ be a fixed integer and suppose that we want to estimate $\langle Z^{(i)} \rangle_l$ at a sufficiently large time $l$ (that depends on $\epsilon$, i.e., the rate of forgetting the initial condition). Suppose we initialize $N_m$ circuits in $|0\rangle^{\otimes n}$, re-apply and re-run $T(u_k)$ from $k=1$ as before. We then obtain the QR states $\rho_{l-M}$ at time $l-M$ and $\rho_l$ at time $l$. Thanks to the convergence property, we can instead re-initialize the $N_m$ circuits in $|0\rangle^{\otimes n}$ at time $l-M$ and from this time onwards re-apply and re-run $T(u_k)$ according to inputs $\{u_{l-M+1}, \ldots, u_l\}$. At time $l$, we have the corresponding QR state $\tilde{\rho}_l$. By the convergence property (see Appendix~\ref{app-sec:qnd} for the derivation), we can make the difference between $\rho_l$ and $\tilde{\rho}_l$ negligible by choosing $M$ appropriately based on $\epsilon$. If we perform repeated measurements on $\rho_l$ and $\tilde{\rho}_l$, the estimates of $\langle Z^{(i)} \rangle_l$ and $\langle \widetilde{Z^{(i)}} \rangle_l = {\rm Tr}(\tilde{\rho}_l Z^{(i)})$ will also be close; see Appendix~\ref{app-sec:mc-estimate}.

The convergence property can be readily exploited on current NISQ machines, leading to efficient versions of both schemes. The first scheme now requires $N_m S L$ circuit runs but only $N_m S M$ applications of $T(u_l)$. The second scheme now only needs $N_m S$ circuit runs and $N_m S M$ applications of $T(u_l)$, both are \textit{independent} of the input length $L$, enabling QR's potential for fast and scalable temporal processing. In all schemes, it is possible and perhaps advantageous to set $S=1$ and run $N_m$ circuits (possibly in parallel if multiple copies of the same hardware are available), for a sufficiently large $N_m$. The average of $N_m$ measurements estimates $\langle Z^{(i)} \rangle$, whose estimation accuracy increases as $N_m$ increases; see Appendix~\ref{app-sec:mc-estimate} for the analysis. Since qubit reset is not yet available on the IBM superconducting quantum devices, we employ the first implementation scheme in our proof-of-principle experiments. It will be a future work of interest to realize these more efficient protocols on gate-model quantum hardware.

\section{Proof-of-principle experiments} 
\label{sec:exp}
Five nonlinear tasks are chosen to carefully test different computational aspects of the QR proposal. Tasks~I-IV have the fading memory property. Tasks~\Romannum{1} and \Romannum{2} test the QR's ability to learn high-dimensional nonlinear maps. Both tasks are governed by linear dynamics determined by some matrix $A$ and have the same form of nonlinear output. The maximum singular value $\sigma_{\max}(A)$ determines the rate at which the dynamics forgets its initial condition while the sparsity of $A$ reflects the pairwise correlation of the reservoir state elements. Task~\Romannum{1} is described by a dense matrix $A$ with $\sigma_{\max}(A) = 0.5$ and Task~\Romannum{2} is governed by $A$ with $95\%$ sparsity with $\sigma_{\max}(A) = 0.99$. Task~\Romannum{3} tests the QR's ability to learn nonlinear maps governed by a highly nonlinear dynamics. Task~\Romannum{4} tests the short-term memory ability and Task~\Romannum{5} is a long-term memory map for testing the capability of the QR beyond its theoretical guarantee. For all experimental and numerical details, see Appendix~\ref{app-sec:exp}. 

We implement four distinct QRs from the subclass described in Sec.~\ref{subsec:subclass} on three IBM superconducting quantum processors \cite{IBMQ}. Each QR consists of a single subsystem ($N=1$ in Eq.~\eqref{eq:dynamics}) with a linear output function ($R=1$ in Eq.~\eqref{eq:readout}). Hereafter, we drop the subsystem index $k$. A 4-qubit and a 10-qubit QRs are implemented on the 20-qubit Boeblingen device; qubits with lower gate errors and longer coherence times are chosen. The 5-qubit Ourense and Vigo devices are used for two distinct 5-qubit QRs. These 5-qubit quantum devices admit simpler qubit couplings but lower gate errors than the 20-qubit Boeblingen device; see Appendix~\ref{app-subsec:hardware} for hardware specifications. Through comparison among the four QRs, we can investigate the impact of the size of QRs, the complexity of quantum circuits implementing the QR dynamics and the intrinsic hardware noise on the QRs' approximation performance.

\subsection{Quantum circuits for QRs} \label{subsec:qr-circuits}
We require the QRs to forget initial conditions for approximating fading memory maps. Traditionally, initial conditions are washed-out with a sufficiently long input sequence until reaching a steady state. Here we bypass the washout by choosing $\sigma = (|0 \rangle \langle 0|)^{\otimes n}$ and $U_0$ so that $|0\rangle^{\otimes n}$ is the steady state of Eq.~\eqref{eq:dynamics} under $u_l=1$, meaning that we can initialize the QR circuits in $|0 \rangle^{\otimes n}$. Furthermore, $U_0$ and $U_1$ should be different and hardware-efficient but sufficiently complex to produce non-trivial quantum dynamics. We choose a circuit schematics (also see Fig.~\ref{fig:2}(a) and (b)),
\begin{equation} \label{eq:qr-circuit}
\begin{split}
U_0(\boldsymbol{\theta}) & = \prod_{j=1}^{N_0} \left( U^{(j_t)}_3(\boldsymbol{\theta}_{j_t}) {\rm CX}_{j_cj_t} U^{(j_t)}_3(\boldsymbol{\theta}_{j_t})^\dagger \right), \\
U_1(\boldsymbol{\phi}) & = \bigotimes_{i=1}^{n} U^{(i)}_3(\boldsymbol{\phi}_{0_i})  \prod_{j=1}^{N_1}\left(\bigotimes_{i=1}^{n} U^{(i)}_3(\boldsymbol{\phi}_{j_i}) {\rm CX}_{j_cj_t} \right),
\end{split}
\end{equation}
where $\boldsymbol{\theta}_{j_t} = (\boldsymbol{\theta}^0_{j_t}, \boldsymbol{\theta}^1_{j_t}, \boldsymbol{\theta}^2_{j_t})$ and $\boldsymbol{\phi}_{j_i} = (\boldsymbol{\phi}^0_{j_i}, \boldsymbol{\phi}^1_{j_i}, \boldsymbol{\phi}^2_{j_i})$ are gate parameters, each independently and uniformly randomly sampled from $[-2\pi, 2\pi]$. Here $U^{(i)}_3$ is an arbitrary rotation on single qubit $i$ \cite{cross2017open} with inverse $U^{(j_t)}_3(\boldsymbol{\theta}_{j_t})^\dagger=U^{(j_t)}_3(-\boldsymbol{\theta}^0_{j_t}, -\boldsymbol{\theta}^2_{j_t}, -\boldsymbol{\theta}^1_{j_t})$, and ${\rm CX}_{j_c j_t}$ is the CNOT gate with control qubit $j_c$ and target qubit $j_t$. These quantum gates are native to the aforementioned IBM superconducting quantum processors, meaning that no further decomposition into simpler gates is required to implement these chosen gates  \cite{IBMQ}. The numbers of layers $N_0$ and $N_1$ are sufficiently large to couple all qubits linearly while respecting the coherence limits of these devices. Owing to the more flexible qubit couplings in the Boeblingen device, circuits implementing the 4-qubit and 10-qubit QRs have more gate and random parameters than the 5-qubit QRs'.

\begin{figure}[!ht]
\centering
\includegraphics[scale=1, trim={0 0 0 0}, clip]{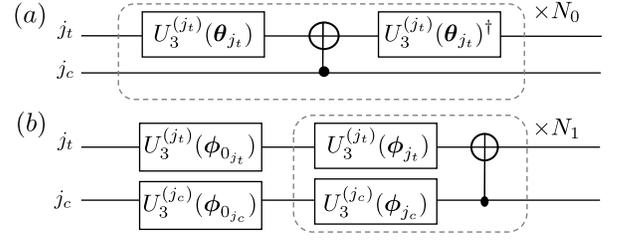}
\caption{Quantum circuit schematics for (a) $U_0(\boldsymbol{\theta})$ and (b) $U_1(\boldsymbol{\phi})$ employed in proof-of-principle experiments, described by Eq.~\eqref{eq:qr-circuit} in Sec.~\ref{subsec:qr-circuits}. Here $j_t$ and $j_c$ are the target and control qubits, respectively. The unitaries $U_0(\boldsymbol{\theta}), U_1(\boldsymbol{\phi})$ consist of $N_0, N_1$ layers of highlighted gate operations, with each layer acting on a different qubit pair ($j_t, j_c$).}
\label{fig:2}
\end{figure}

\begin{figure}[!ht]
\centering
\includegraphics[scale=1]{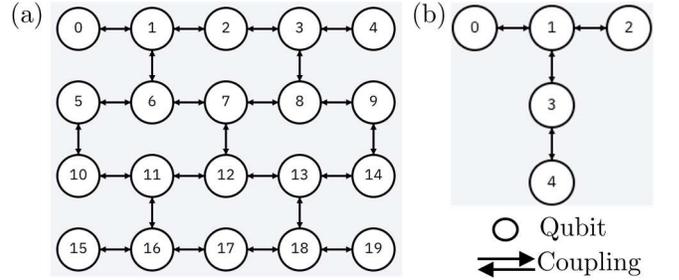}
\caption{Qubit coupling maps of the IBM superconducting quantum processors. (a) The 20-qubit Boeblingen device. (b) Both the 5-qubit Ourense and Vigo devices.}
\label{fig:3}
\end{figure}

For the 4-qubit and 10-qubit QRs on the Boeblingen device, we choose the number of layers $N_0 = N_1 = 5$ in Eq.~\eqref{eq:qr-circuit}. For the 5-qubit Ourense QR, we implement a simpler form of Eq.~\eqref{eq:qr-circuit}, given by
$$U_0 = \prod_{j=1}^{4} {\rm CX}_{j_cj_t}, \quad U_1(\boldsymbol{\phi}) = \bigotimes_{i=1}^{5} U^{(i)}_3(\boldsymbol{\phi}_{i}).$$

To implement a different QR dynamics on the 5-qubit Vigo device, we choose
\begin{equation*}
\begin{split}
& U_0(\boldsymbol{\theta}) = \prod_{j=1}^{3} \left( R^{(j_t)}_y(\boldsymbol{\theta}_{j_t}) {\rm CX}_{j_c j_t} R^{(j_t)}_{y}(\boldsymbol{\theta}_{j_t})^\dagger \right), \\
& U_1(\boldsymbol{\phi}) = \bigotimes_{i=1}^{5} R^{(i)}_x(\boldsymbol{\phi}_i).
\end{split}
\end{equation*}
Here $R^{(i)}_y$ and $R^{(i)}_x$ are rotational $Y$ and $X$ gates on qubit $i$, respectively. Both gates are special instances of the arbitrary single-qubit rotational gate $U^{(i)}_3$ with one (free) gate parameter while the other two being fixed constants. For all QRs, natively coupled control and target qubits for the CNOT gates are chosen, meaning that a CNOT gate can be directly applied to the qubit pair without additional gate operations. See Fig.~\ref{fig:3} for the device qubit coupling maps and Appendix~\ref{app-subsec:circuits} for the QR quantum circuit details.

\subsection{Experimental implementation}
\label{subsec:exp-im}
In this section we report on experiments demonstrating the first implementation scheme described in Sec.~\ref{sec:realization}. We choose a sufficiently large $N_m=1024$ and $\epsilon = 0.1$ for a moderate short-term memory. To estimate $\langle Z^{(i)} \rangle_l$ at time $l$, each of the $N_m$ circuits implementing the QRs on the Boeblingen device and the 5-qubit QRs are run for $S=1024$ and $S=8192$ shots, respectively. These shot numbers are chosen according to circuit execution times of the devices. 

\begin{figure*}[!ht]
\includegraphics[scale=1, trim={0 0 0 0}, clip]{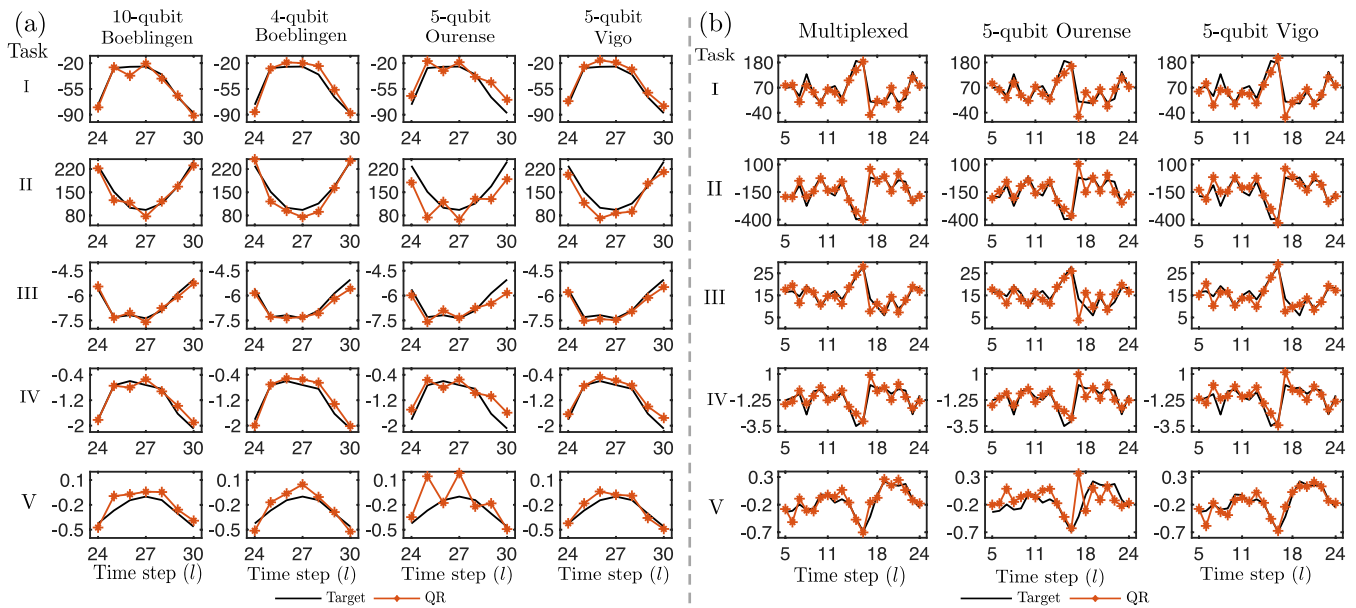}
\caption{(a) Shows the QRs' predicted outputs for the multi-step prediction problem, rows and columns correspond to different tasks and QRs, respectively. (d) Shows the QRs' predicted outputs for the map emulation problem, first column corresponds to the multiplexed QR.}
\label{fig:4}
\end{figure*}

We apply the four QRs to the five nonlinear tasks on the multi-step ahead prediction and map emulation problems. To implement the same washout as for the QRs for each target map, we inject a constant input sequence $u_l=1$ of length $50$ followed by train and test inputs uniformly randomly sampled from $u_l \in [0,1]$. This change in the input statistics leads to a transitory target output response. We remove the associated transients by discarding the first four target input-output data and the corresponding QR experimental data, see Appendix~\ref{app-subsec:data} for all data. For the multi-step ahead problem, train and test time steps run from $l=5$ to $L_T=23$ and $L_T+1=24 $ to $L=30$, respectively. For the map emulation problem, $K=2$ train input-output pairs running from $l=5$ to $L'=24$ are used, followed by one unseen test input-output pair with the same time steps. The number of train and test data in our proof-of-principle experiments is limited by the length of quantum circuits allowed on the IBM quantum processors. Furthermore, these cloud-based quantum processors are shared among users, making continuous experiments infeasible and durations of experiments lengthy. Yet our work indicates that despite these current limitations, NISQ devices can demonstrate learning of input-output maps and supports QR as a viable intermediate application of NISQ machines on the road to full-fledged quantum devices equipped with quantum error correction.

To harness the flexibility of the QR approach, a multi-tasking technique is used, in which the four QRs are evolved and the estimates of $\langle Z^{(i)} \rangle_l$ for all time steps are recorded once, whereas the readout parameters $\textbf{w}$ are optimized independently for each task. That is a fixed QR dynamics, with fixed gate parameter values, is exploited for multiple tasks simultaneously. We evaluate and compare the task performance of QRs using the normalized mean-squared error between prediction $\overline{y}\rvert_{L_T+1:L}$ and target $y\rvert_{L_T+1:L}$, computed as 
\vspace*{-0.5em}
$${\rm NMSE} = \sum_{l=L_T+1}^{L} |y_l - \overline{y}_l|^2  /  \Delta_y^2, \vspace*{-0.5em}$$
where $\mu = \frac{1}{L-L_T} \sum_{l=L_T+1}^{L} y_l$, $\Delta^2_y = \sum_{l=L_T+1}^{L}(y_l - \mu)^2$. While the success of experimental demonstration of hybrid quantum-classical algorithms often requires error mitigation techniques to reduce the effect of decoherence \cite{kandala2019error,li2017efficient}, we remark that our results are obtained without any process or readout error mitigation.

\subsection{QR task performance}
As the number of qubits increases, the 10-qubit Boeblingen QR is expected to perform better than other QRs. For the multi-step ahead prediction problem, we observe that two qubits in the 10-qubit Boeblingen QR experienced significant time-varying deviations between the experimental data and simulation results on the Qiskit simulator; see Appendix~\ref{app-subsec:measurement-data} for a discussion. To remedy this issue, {we set the corresponding elements of $\textbf{w}$ to be zeros. The resulting 10-qubit Boeblingen QR (with NMSE$<$0.08) outperforms other QRs with a smaller number of qubits on the first four tasks, and achieves an almost two-fold performance improvement on Tasks~\Romannum{2} and \Romannum{3}; see Table~\ref{table:multi-step-ahead} for all NMSEs on the multi-step ahead prediction problem. The 10-qubit Boeblingen QR predicted outputs follow the target outputs relatively closely as shown in Fig.~\ref{fig:4}(a). The 5-qubit Ourense QR admits very simple dynamics, whereas the 5-qubit Vigo QR has more gate operations and gate parameters. The 5-qubit Ourense QR is outperformed by the 5-qubit Vigo QR in all tasks. Considering that the Ourense and Vigo devices have similar noise characteristics and the same qubit coupling map, this suggests the QR performance can be improved by choosing a more complex quantum circuit, in the sense of having a longer gate sequence.

\begin{table}[!ht]
\centering
\small
\caption{NMSEs on the multi-step ahead prediction.}
\label{table:multi-step-ahead}
\scalebox{1}{
\setlength{\tabcolsep}{7pt}
\begin{tabular}{ c c c c c  } 
Task         & 10-qubit   & 4-qubit    & 5-qubit & 5-qubit  \\
             & Boeblingen & Boeblingen & Ourense & Vigo     \\ 
\hline
\Romannum{1} & 0.051     &  0.088      & 0.24    & 0.070        \\
\Romannum{2} & 0.072     &  0.12       & 0.68    & 0.22       \\
\Romannum{3} & 0.043     &  0.10       & 0.25    & 0.081        \\
\Romannum{4} & 0.079     &  0.092      & 0.34    & 0.11         \\
\Romannum{5} & 0.47      &  0.41       & 2.3     & 0.20         \\
\end{tabular} 
}
\end{table}

\begin{table}[!ht]
\centering
\small
\caption{NMSEs on the map emulation.}
\label{table:map-emulation}
\scalebox{1}{
\setlength{\tabcolsep}{11.6pt}
\begin{tabular}{  c  c c c } 
Task         & Multiplexed   & 5-qubit    & 5-qubit  \\
             & QR            & Ourense    & Vigo     \\ 
\hline
\Romannum{1} & 0.20          &  0.26      & 0.32     \\
\Romannum{2} & 0.13          &  0.27      & 0.23     \\
\Romannum{3} & 0.16          &  0.46      & 0.26     \\
\Romannum{4} & 0.25          &  0.30      & 0.36     \\
\Romannum{5} & 0.20          &  1.1       & 0.17     \\
\end{tabular} 
}
\end{table}

\begin{figure}[!ht]
\includegraphics[scale=1, trim={0 0 0 0}, clip]{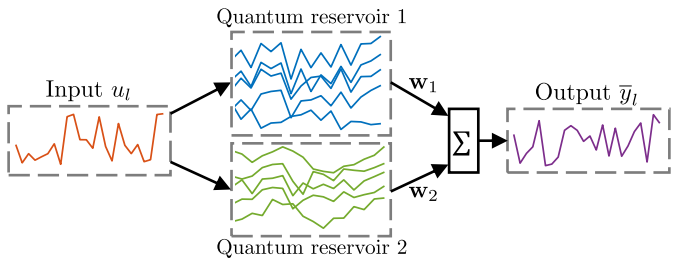}
\caption{The spatial multiplexing schematic. The same input sequence is injected into two distinct 5-qubit QRs. The internal states ${\rm Tr}(\rho_l Z^{(i)})$ of the two QRs are linearly combined to form a single output.}
\label{fig:5}
\end{figure}

The 10-qubit Boeblingen QR performs better on all tasks than the 5-qubit QRs except on Task~\Romannum{5}. This could be due to the impact of the higher noise level in the Boeblingen device and the fact that the output sequence is generated by a map that is not known to be fading memory, see Appendix~\ref{app-subsec:hardware} for the hardware specifications. Our universal class of QRs can exploit the property of spatial multiplexing as initially proposed in Ref.~\cite{nakajima2019boosting}; also see \cite{chen2019learning} and Fig.~\ref{fig:5} for an illustration. Outputs of distinct and non-interacting 5-qubit QRs can be combined linearly to harness the computational features of both members. Since the combined Ourense and Vigo devices have 10 qubits overall as with the 10-qubit Boeblingen QR but with lower noise levels, it would be meaningful to combine the 5-qubit Vigo and Ourense QRs via spatial multiplexing on the map emulation problem. The results of this multiplexing is summarized in Table~\ref{table:map-emulation}.

The combination of two 5-qubit QRs as discussed above achieves ${\rm NMSE}=0.20, 0.13, 0.16, 0.25, 0.20$ for the five tasks without any readout or process error mitigation. The predicted multiplexed QR outputs corresponding to the unseen inputs follow the target outputs relatively closely as shown in Fig.~\ref{fig:4}(b). Without spatial multiplexing, the 5-qubit Ourense or the 5-qubit Vigo QR show a worse performance in the first four tasks; see Table~\ref{table:map-emulation}. The spatial multiplexed 5-qubit QR combines computational features from the constituent QRs and can achieve comparable performance to the individual members as well as gaining an almost two-fold performance boost on Tasks~\Romannum{2} and \Romannum{3}. We anticipate that spatial multiplexing of QRs with more complex circuit structures and a larger number of qubits can lead to further performance improvements.

\section{Conclusion}
\label{sec:conclusion}
We propose a novel class of quantum reservoir computers endowed with universality property that is implementable on available noisy gate-model quantum hardware for temporal information processing. Our approach can harness arbitrary but fixed quantum circuits native to noisy quantum processors, without precise tuning of the circuit parameters. Our theoretical analysis is supported by proof-of-concept experiments on current superconducting quantum devices, demonstrating that small-scale noisy quantum reservoirs can perform non-trivial nonlinear temporal processing tasks under current hardware limitations, in the absence of readout and process error mitigation techniques. We also detail more efficient implementation schemes of our QR proposal that could enable QR's potential for fast and scalable temporal processing. It is a future work of interest to realize these more efficient protocols on quantum hardware. Our work indicates that quantum reservoir computing can serve as a viable intermediate application of NISQ devices on the road to full-fledged quantum computers.

Our approach is scalable in the number of qubits by offloading exponentially costly computations to noisy quantum systems and utilizing classical algorithms with a linear (in the number of qubits) computational cost to process sequential data. Moreover, when implemented on NISQ devices, the micro-second timescale for the evolution of the quantum reservoir suggests its potential for real-time fast signal processing tasks. Guided by our theory, we applied the spatial multiplexing technique initially proposed in \cite{nakajima2019boosting}, and demonstrate experimentally that exploiting distinct computational features of multiple small noisy quantum reservoirs can lead to a computational boost. As NISQ hardware becomes increasingly accessible and the noise level is continually reduced, we anticipate that the quantum reservoir approach will find useful applications in a broad range of scientific disciplines that employ time series modeling and analysis. We are also optimistic for useful applications to be possible even with a noise level above the threshold for continuous quantum error correction.

\section{Acknowledgments}
The authors thank Keisuke Fujii for an insightful discussion. NY is supported by the MEXT Quantum Leap Flagship Program Grant Number JPMXS0118067285.

\appendix
\section{Universality for approximating nonlinear fading memory maps}
\label{app-sec:universality}
We first define notation for the rest of this section. Let $K([0, 1])$ be the set of infinite sequences $u = \{ \ldots, u_{-1}, u_{0}, u_{1}, \ldots \}$ such that $u_l \in [0, 1]$ for all $l \in \mathbb{Z}$. Let $K^{+}([0, 1])$ and $K^{-}([0, 1])$ be subsets of $K([0, 1])$ for which the indices are restricted to $\mathbb{Z}^{+} = \{1, 2, \ldots\}$ and $\mathbb{Z}^{-} = \{\ldots, -2, -1, 0\}$, respectively. For any complex matrix $A$, $\| A \|_{p} = {\rm Tr}(\sqrt{A^\dagger A}^{p})^{1/p}$ is the Schatten $p$-norm for some $p \in [1, \infty)$. For any operator $T$, the induced operator norm is $\|T\|_{p-p} = \sup_{A \in \mathbb{C}^{n \times n}, \|A\|_p =1} \| T(A) \|_p$. Let $\mathcal{D}(2^n)$ denotes the set of $2^n \times 2^n$ density operators.

Consider an input-output map $M$ that maps an infinite input sequence $u \in K([0, 1])$ to a real infinite output sequence $y \in K(\mathbb{R})$. We say that $M$ is $w$-fading memory if there exists a decreasing sequence $w = \{w_0, w_1, \ldots\}$ with $\lim_{l \rightarrow \infty} w_{l} = 0$, such that for any $u ,v \in K^{-}([0, 1])$, we have $|M(u)_0 - M(v)_0| \rightarrow 0$ whenever $\sup_{l \in \mathbb{Z}^{-}}|w_{-l}(u_l - v_l)|  \rightarrow 0$. Here $M(u)_l = y_l$ is the output sequence at time $l$. We also require $M$ to be causal and time-invariant as in Ref.~\cite{chen2019learning}, meaning that the output of $M$ at time $l$ only depends on the input up to and including time $l$, and its outputs are invariant under time-shifts. Now we are interested in approximating $M$ with a time-invariant fading memory map $\overline{M}$ produced by a quantum reservoir computer.

\subsection{The convergence property}
\label{app-subsec:cv}
Since $M$ is fading memory, the map $\overline{M}$ must also forget its initial condition $\rho_0$. This is the convergence property \cite{PWN05} or the echo-state property \cite{JH04}. We give a precise definition here.
\begin{definition}[Convergence]
An input-dependent CPTP map $T$ is convergent with respect to input $u \in K([0, 1])$ if there exists a sequence $\{\delta_l; l \geq 0\}$ with $\delta_l > 0$ and $\lim_{l \rightarrow \infty} \delta_l = 0$ such that for all $u \in K^{+}([0, 1])$, for any two density operators $\rho_{j, l}$ $(j= 1, 2)$ satisfying $\rho_{j, l} = T(u_l) \rho_{j, l-1}$, it holds that $\| \rho_{1, l} - \rho_{2, l} \|_1 \leq \delta_l$. If a QR dynamic $T$ is convergent, we call the QR a convergent system.
\end{definition}

\begin{lemma}
The QR dynamics given by Eqs.~\eqref{eq:dynamics} and \eqref{eq:dynamics-subsystem} is convergent with respect to inputs $u \in K([0,1])$.
\end{lemma}
\begin{proof}
First we show each subsystem governed by Eq.~\eqref{eq:dynamics-subsystem} is convergent. For any $\rho, \sigma \in \mathcal{D}(2^n)$, $u_l \in [0, 1]$ and $\epsilon_k \in (0, 1]$, we have
\begin{equation}
\label{eq:cv}
\begin{split}
& \quad  \| T^{(k)}(u_l)(\rho - \sigma) \|_{1} \\
& = (1-\epsilon_k) \left\| \left(u_l T^{(k)}_0 + (1-u_l)T^{(k)}_1 \right) (\rho -\sigma) \right\|_1 \\
& \leq (1-\epsilon_k) \| \rho -\sigma\|_1  \leq 2 (1-\epsilon_k) ,
\end{split}
\end{equation}
where the first inequality follows from \cite[Theorem 9.2]{NC10} and the convex combination $u_l T^{(k)}_0 + (1-u_l)T^{(k)}_1$ is again a CPTP map. Now let $\rho_{1, 0}$ and $\rho_{2, 0}$ be two arbitrary initial density operators, using the inequality Eq.~\eqref{eq:cv} $L$ times, we have
\begin{equation*}
\begin{split}
& \| \rho_{1, L} - \rho_{2, L} \|_1 \\
& = \left\| \left( \overleftarrow{\prod}_{l=1}^{L} T^{(k)}(u_l) \right) (\rho_{1, 0} - \rho_{2, 0})  \right\|_1 \\
& \leq (1- \epsilon_k)^L \left\| \rho_{1, 0} - \rho_{2, 0} \right\|_1 \leq 2(1-\epsilon_k)^L,
\end{split}
\end{equation*}
where $\overleftarrow{\prod}_{l=1}^{L} T^{(k)}(u_l)$ is the time-composition of $T^{(k)}(u_l)$ from right to left. 

Secondly, we show that the QR dynamics Eq.~\eqref{eq:dynamics} is convergent by showing that $T(u_l) = \bigotimes_{k=1}^{N} T^{(k)}(u_l)$ is again convergent when the subsystems are initialized in a product state. We apply the same argument as in \cite[Lemma 5]{chen2019correction}. Consider two CPTP maps $T^{(1)}(u_l)$ and $T^{(2)}(u_l)$ of the form Eq.~\eqref{eq:dynamics-subsystem}. Let $\rho_{1, 0} \otimes \sigma_{1, 0}$ and $\rho_{2, 0} \otimes \sigma_{2, 0}$ be two arbitrary initial product states. Then $T^{(1)}(u_l) \otimes T^{(2)}(u_l)$ is again convergent with respect to all $u \in K([0, 1])$, as shown in the following,
\begin{equation*}
\begin{split}
& \ \| \rho_{1, L} \otimes \sigma_{1, L} - \rho_{2, L} \otimes \sigma_{2, L} \|_{1} \\
& \hspace*{-1em} \leq \left\| \left(\overleftarrow{\prod}_{l=1}^{L} T^{(1)}(u_l) \otimes T^{(2)}(u_l) \right) (\rho_{1, 0} \otimes \sigma_{1, 0} - \rho_{2, 0} \otimes \sigma_{1, 0}) \right\|_{1}  \\
& \hspace*{-0.7em} + \left\|\left(  \overleftarrow{\prod}_{l=1}^{L} T^{(1)}(u_l) \otimes T^{(2)}(u_l) \right) (\rho_{2, 0} \otimes \sigma_{1, 0} - \rho_{2, 0} \otimes \sigma_{2, 0}) \right\|_1 \\
& \hspace*{-1em} = \left\| \left( \overleftarrow{\prod}_{l=1}^{L}  T^{(1)}(u_l) \right)(\rho_{1, 0} - \rho_{2, 0}) \right\|_{1} \left\| \sigma_{1, L} \right\|_{1} \\
& \hspace*{-0.7em} + \left\| \left( \overleftarrow{\prod}_{l=1}^{L}  T^{(2)}(u_l) \right) (\sigma_{1, 0} - \sigma_{2, 0}) \right\|_1  \left\| \rho_{2, L} \right\|_1 \\ 
& \hspace*{-1em} \leq 2(1-\epsilon_1)^L + 2(1-\epsilon_2)^L.
\end{split}
\end{equation*}
Repeating this argument $N$ times shows the QR dynamics $T(u_l) = \bigotimes_{k=1}^{N} T^{(k)}(u_l)$ is again convergent.
\end{proof}

\subsection{The fading memory property}
\label{app-subsec:fmp}
Associate the readout function Eq.~\eqref{eq:readout} to the QR dynamics Eqs.~\eqref{eq:dynamics} and \eqref{eq:dynamics-subsystem}. This defines an input-output (I/O) map $\overline{M}_{(T, h_{\textbf{w}})}$.This I/O map is causal, meaning that the its output $\overline{y}_l$ only depends on $u_{l'}$ for $l' \leq l$. Furthermore, it is time-invariant, meaning that TODO:DEFN HERE.

$\overline{y}_{\tau + l} = \overline{M}_{(T, h_\textbf{w})}(S_\tau(u))_l$ for all $\tau \in \mathbb{Z}$, where $S_\tau(u) = \{\ldots, u_{\tau-1}, u_{\tau}, u_{\tau+1} \ldots\}$ shifts the input sequence by $\tau$. By causality and time-invariance, it suffices to consider the outputs $\overline{y}_l$ of $\overline{M}_{(T, h_\textbf{w})}(u)_l$ for $l \leq 0$ and left-infinite inputs $u \in K^{-}([0, 1])$; see \cite{BC85,grigoryeva2018universal,chen2019learning} for details.

For any $u \in K^{-}([0, 1])$ and any initial condition $\rho_{-\infty}$,
\begin{equation*}
\overline{M}_{(T, h_{\textbf{w}})}(u)_0 = h_\textbf{w} \left( \left( \overrightarrow{\prod}_{j=0}^{\infty} T(u_{-j})\right) \rho_{-\infty}  \right),
\end{equation*}
where $\overrightarrow{\prod}_{j=0}^{\infty} T(u_{-j}) = \lim_{N \rightarrow \infty} T(u_0) \cdots T(u_{l-N})$ and the limit is point-wise. We can restate the fading memory property in terms of continuity of $\overline{M}_{(T, h_\textbf{w})}$ with respect to a certain norm. Given a null sequence $w$ (i.e., a decreasing sequence $w$ with $\lim_{l \rightarrow \infty} w_l = 0$) and any $u \in K^{-}([0, 1])$, define a weighted norm $\| u \|_{w} = \sup_{l \in \mathbb{Z}^{-}}|u_{l}|w_{-l}$. The map $\overline{M}_{(T, h_\textbf{w})}$ is $w$-fading memory if it is continuous in $(K^{-}([0, 1]), \| \cdot \|_w)$.

\begin{definition}[Fading memory]
Given a null sequence $w$, the set of $w$-fading memory maps is the set of all continuous functions $C(K^{-}([0, 1]), \| \cdot \|_w)$ defined on $(K^{-}([0, 1]), \| \cdot \|_w)$.
\end{definition}

\begin{lemma}
For any null sequence $w$, $\overline{M}_{(T, h_{\textbf{w}})}$ induced by QR described by Eqs.~\eqref{eq:dynamics}--\eqref{eq:readout} is $w$-fading memory.
\end{lemma}
\begin{proof}
Using the same argument in \cite[Lemma~3]{chen2019learning}, it follows that $\overline{M}_{(T, h_{\textbf{w}})}$ is $w$-fading memory if each $k$-th subsystem dynamics $T^{(k)}(u_l)$ is continuous with respect to the inputs $u_l \in [0, 1]$ for all $k=1, \ldots, N$. If fact, we show that $T^{(k)}(u_l)$ is uniformly continuous. Let $x, y \in [0, 1]$ and $A \in \mathbb{C}^{2^{n_k} \times 2^{n_k}}$,
\begin{equation*}
\vspace*{0.1em}
\begin{split}
& \hspace{1.5em} \| T^{(k)}(x) - T^{(k)}(y) \|_{1-1}  \\
& = \sup_{A \in \mathbb{C}^{2^{n_k} \times 2^{n_k}}, \|A\|_1 = 1} \left\|\left(T^{(k)}(x) - T^{(k)}(y)\right) A \right\|_1 \\
& = (1-\epsilon_k)|x-y| \sup_{A \in \mathbb{C}^{2^{n_k} \times 2^{n_k}}, \|A\|_1 = 1} \left\|T_0^{(k)}(A) - T_{1}^{(k)}(A) \right\|_1 \\ 
& \leq (1-\epsilon_k) |x-y| \left( \left\| T^{(k)}_0 \right\|_{1-1}  + \left\|T^{(k)}_1 \right\|_{1-1} \right) \\
& \leq 2(1-\epsilon_k)|x-y|,
\end{split}
\end{equation*}
where the last inequality follows from \cite[Theorem~2.1]{perez2006contractivity}. We remark that \cite[Lemma~3]{chen2019learning} is stated with respect to the Schatten $p=2$ norm, but the same argument holds for Schatten $p=1$ norm.
\end{proof}

\subsection{The universality property}
\label{app-subsec:universal}
Now consider the family $\mathcal{M}$ of maps $\overline{M}_{(T, h_{\textbf{w}})}$. We state our main universality result.
\begin{theorem}[Universality] \label{theorem:universal}
For any null sequence $w$, the QR class $\mathcal{M}$ is dense in $C(K^{-}([0, 1]), \| \cdot \|_w)$. That is, given any $w$-fading memory map $M \in C(K^{-}([0, 1]), \| \cdot \|_w)$ and any $\delta > 0$, there exists $\overline{M}_{(T, h_\textbf{w})} \in \mathcal{M}$ such that for all $u \in K^{-}([0, 1])$, $\sup_{l \in \mathbb{Z}^{-}} |M(u)_l - \overline{M}_{(T, h_\textbf{w})}(u)_l | < \delta$.
\end{theorem}
We apply the Stone-Weierstrass Theorem to show that $\mathcal{M}$ is dense in $C(K^{-}([0, 1]), \| \cdot \|_{w})$ . It has been shown that the space $(K^{-}([0, 1]), \| \cdot \|_w)$ is a compact metric space \cite[Lemma~2]{grigoryeva2018universal}. We now state the Stone-Weierstrass Theorem.

\begin{theorem}[Stone-Weierstrass]
Let $E$ be a compact metric space and $C(E)$ be the set of real-valued continuous functions defined on $E$. If a subalgebra $A$ of $C(E)$ contains the constant functions and separates points of $E$, then $A$ is dense in $C(E)$.
\end{theorem}

\begin{proof}[Proof of Theorem~\ref{theorem:universal}]
The family $\mathcal{M}$ forms a polynomial algebra follows from \cite[Lemma~5]{chen2019learning} and the observation that for any QR dynamics $T_1(u_l) = \bigotimes_{k=1}^{N_1} T_1^{(k)}(u_l)$ and $T_2(u_l) = \bigotimes_{k=1}^{N_2} T_2^{(k)}(u_l)$, where each $T_1^{(k)}, T^{(k)}_2$ has the form Eq.~\eqref{eq:dynamics-subsystem}, we again have $T(u_l) (\rho_1 \otimes \rho_2 )= T_1(u_l) \rho_1 \otimes T_2(u_l) \rho_2$ is of the form Eq.~\eqref{eq:dynamics}. Furthermore, $T(u_l) = T_1(u_l) \otimes T_2(u_l)$ is again convergent when initialized in a product state of the subsystems. Therefore, the family $\mathcal{M}$ forms a polynomial algebra consisting of $w$-fading memory maps.
    
Constant functions can be obtained by setting $\textbf{w}_{i_1, \ldots, i_{n}}^{r_{i_1}, \ldots, r_{i_n}}=0$ in Eq.~\eqref{eq:readout}. It remains to show that $\mathcal{M}$ separates points in $K^{-}([0, 1])$. That is, for any distinct $u, v \in K^{-}([0, 1])$ with $u_l \neq v_l$ for at least one $l$, we need to find a map $\overline{M}_{(T, h_{\textbf{w}})} \in \mathcal{M}$ such that $\overline{M}_{(T, h_{\textbf{w}})}(u)_0 \neq \overline{M}_{(T, h_{\textbf{w}})}(v)_0$. We show that we can construct a single-qubit quantum reservoir with this property. 

\begin{widetext}
\begin{equation} \label{eq:T-matrix}
\overline{T}(u_l) = |00 \rangle \langle 00| + (1-\epsilon) 
\begin{pmatrix}
0 & 0 & 0 & 0 \\
\sin^2(2J)(2u_l-1) & \cos^2(2J) & 0 & 0 \\
0 & 0 & \cos(2J)\cos(2\alpha) & -\cos(2J)\sin(2\alpha) \\
0 & 0 & \cos(2J)\sin(2\alpha) & \cos(2J) \cos(2\alpha)
\end{pmatrix}
\end{equation}
\end{widetext}

Consider a single-qubit quantum reservoir with a linear readout function ($n=1, R=1, N=1$). For the rest of this proof, we drop the subsystem index. This quantum reservoir consists of one system qubit and one ancilla qubit denoted as $\rho_a$. Choose the dynamics
\begin{equation}
\label{supp_eq:sep}
\begin{split}
\hspace*{-0.7em} \rho_{l} 
& = T(u_l)\rho_{l-1} \\
&  = (1-\epsilon) \left(u_l {\rm Tr}_{a} \left(e^{-i H} (\rho_{l-1} \otimes \rho^0_a) e^{i H} \right) \right. \\
& \qquad \left. + (1-u_l) {\rm Tr}_{a} \left(e^{-iH}  (\rho_{l-1} \otimes \rho^1_a) e^{iH}\right) \right) + \epsilon K_{\frac{I}{2}},
\end{split}
\end{equation}
where $\rho^j_a = |j \rangle \langle j|$ for $j=0, 1$, ${\rm Tr}_a$ denotes the partial trace over ancilla $\rho_a$ and $\epsilon \in (0, 1)$. The map $K_{\frac{I}{2}}$ is a CPTP map defined as $K_{\frac{I}{2}}(X) = {\rm Tr}(X) \frac{I}{2}$ for any $X \in \mathbb{C}^{2 \times 2}$. The Hamiltonian $H$ is of the Ising type $H = J(X^{(0)} X^{(1)} + Y^{(0)} Y^{(1)}) + \alpha \sum_{j=0}^{1} Z^{(j)}$, where $X^{(j)}, Y^{(j)}$ and $Z^{(j)}$ are the Pauli $X, Y$ and $Z$ operators on qubit $j$, with $j=0$ being the ancilla qubit.

We order an orthogonal basis for $\mathbb{C}^{2 \times 2}$ as $\{I, Z, X, Y\}$. The matrix representation of the CPTP map Eq.~\eqref{supp_eq:sep} is given by Eq.~\eqref{eq:T-matrix}. Since Eq.~\eqref{supp_eq:sep} is convergent, we can choose any initial condition $\rho_{-\infty} = |0 \rangle \langle 0|$ with the corresponding vector representation $\overline{\rho}_{\infty} = \frac{1}{2}\begin{pmatrix}1 & 1 & 0 & 0\end{pmatrix}$. Taking a linear readout function, for $u\in K^{-}([0, 1])$, the quantum reservoir implements
\begin{equation*}
\overline{M}_{(T, h_\textbf{w})}(u)_0 = 2\textbf{w}_1 \left[\left( \overrightarrow{\prod}_{j=0}^{\infty} \overline{T}(u_{-j}) \right) \overline{\rho}_{-\infty} \right]_2 + \textbf{w}_c,
\end{equation*}
where $[ \cdot ]_2$ is the second element of the vector corresponding to ${\rm Tr}(Z \rho_{0})/2$.

Now given two distinct inputs $u, v \in K^{-}([0, 1])$, suppose that $u_0 \neq v_0$. Then choose $J$ such that $\cos^2(2J)=0$ and therefore,
\begin{equation*}
\overline{M}_{(T, h_\textbf{w})}(u)_0 - \overline{M}_{(T, h_\textbf{w})}(v)_0 = 2\textbf{w}_1 (1-\epsilon)(u_0 - v_0) \neq 0.
\end{equation*} 

Suppose $u_0 = v_0$, note that in general
\begin{equation*}
\begin{split}
& \overline{M}_{(T, h_\textbf{w})}(u)_0 \\
& = \textbf{w}_1 \sin^{2}(2J) (1-\epsilon) \sum_{j=0}^{\infty} \left( (1-\epsilon) \cos^2(2J) \right)^j (2u_{-j} - 1).
\end{split}
\end{equation*}
Choose $\epsilon \in (0, 1)$ and $J$ such that $(1-\epsilon)\cos^{2}(2J) \in (0, 1-\epsilon)$. Then the above is a convergent power series and the subtraction is well-defined:

\begin{equation*}
\begin{split}
& \overline{M}_{(T, h_\textbf{w})}(u)_0 - \overline{M}_{(T, h_\textbf{w})}(v)_0 \\
& = 2\textbf{w}_1 \sin^2(2J) (1-\epsilon) \sum_{j=0}^{\infty} \left((1-\epsilon) \cos^2(2J)\right)^j(u_{-j} - v_{-j}).
\end{split}
\end{equation*}

The above is a power series of the form
\begin{equation*}
f(\theta) = 2\textbf{w}_1 \sin^2(2J)(1-\epsilon) \sum_{j=0}^{\infty} \theta^j(u_{-j} - v_{-j}), 
\end{equation*}
where $f(\theta)$ has a nonzero radius of convergence and is non-constant since $\theta = (1-\epsilon)\cos^2(2J) \in (0, 1-\epsilon)$, $(1-\epsilon)\sin^2(2J) \in (0, 1-\epsilon)$ and $u, v$ are assumed to be distinct. Furthermore, since we assume that $u_0 = v_0$, we have $f(0) = 0$. Invoking \cite[Theorem~3.2]{lang1985complex}, there exists $\beta > 0$ such that $f(\theta) \neq 0$ for all $|\theta| \leq \beta, \theta \neq 0$. This concludes the proof for separation of points. The universality of $\mathcal{M}$ now follows from the Stone-Weierstrass Theorem.
\end{proof}

\section{Invariance under time-invariant readout error}
\label{app-sec:readout-error}
The QR outputs are invariant under time-invariant readout error whenever a linear readout function is used. That is when $R=1$ in Eq.~\eqref{eq:readout}, the QR predicted outputs $\overline{y}_l$ remain unchanged under time-invariant readout error. Let $\mathcal{B} =\{ | i \rangle \}$ be the computational basis for an $n$-qubit system, with $i=1, \ldots, 2^n$. The readout error is characterized by a measurement calibration matrix $A$ whose $i,j$-th element $A_{i, j} = {\rm Pr}(i|j)$ is the probability of measuring the state $|i\rangle \in \mathcal{B}$ given that the state is prepared in the state $|j\rangle \in \mathcal{B}$.

We employ the readout error correction method described in Ref.~\cite{havlivcek2019supervised}. For an $n$-qubit QR, at each time step $l$, we execute $2^n$ calibration circuits with each circuit initialized in one of the $2^n$ computational basis elements. The outcomes are used to create the measurement calibration matrix $A_l$. The readout error at time step $l$ is corrected by applying the pseudo-inverse of $A_l$ to the measured outcomes from the experiments.

For all experiments, the measurement outcomes are stored as the count of measuring each basis elements in $\mathcal{B}$. Let $\textbf{v}_l = \begin{pmatrix} \textbf{v}^{1}_l & \cdots & \textbf{v}^{2^n}_l & 1 \end{pmatrix}$, where $\textbf{v}^i_l$ is the count of measuring $| i \rangle \in \mathcal{B}$ at time step $l$. Let $\textbf{z}_l = \begin{pmatrix} \overline{\langle Z^{(1)} \rangle}_l & \cdots & \overline{\langle Z^{(n)} \rangle}_l & 1\end{pmatrix}$, where $\overline{\langle Z^{(i)} \rangle}_l$ is the finite-sampled approximation of $\langle Z^{(i)} \rangle_l$ for $i=1,\ldots, n$. Then we have $\textbf{z}_l = \textbf{v}_l B$, where $B$ is a linear transformation. After applying readout error correction, we have $\textbf{z}'_l = \textbf{v}_l A^{+}_l B$, where $A_l^+$ is the pseudo-inverse of $A_l$. To optimize the readout function parameters $\textbf{w}$, collect all measurement data in a matrix $\textbf{v} = \begin{pmatrix} \textbf{v}_1^\top & \cdots & \textbf{v}_L^\top \end{pmatrix}^\top$ so that $\textbf{z} = \begin{pmatrix} \textbf{z}_1^\top & \cdots & \textbf{z}_L^\top \end{pmatrix}^\top = \textbf{v} B$, where $L$ is the sequence length. The linear output of the quantum reservoir computer is $\overline{\textbf{y}} = \textbf{v} B \textbf{w}$, where $\textbf{w}$ includes the bias term $\textbf{w}_c$. Append a corresponding row and column to $A^\dagger_l$ to account for the bias term. Suppose the readout error is time-invariant, then $A^{+} = A^{+}_l$ for $l = 1, \ldots, L$. The quantum reservoir computer output after readout error correction is $\overline{\textbf{y}'} = \textbf{v} A^+ B \textbf{w}'$. Assume that $A^+$  has all rows linearly independent, then ordinary least squares yields $B\textbf{w}' = A B \textbf{w}$. Now given test data $\textbf{v}_{\rm test}$ with readout error correction applied, $\textbf{v}_{\rm test} A^+ B \textbf{w}' = \textbf{v}_{\rm test} A^+ A B \textbf{w} = \textbf{v}_{\rm test} B \textbf{w}.$
Therefore, the QR predicted outputs are invariant under time-invariant readout error.

\section{Efficient implementations of a subclass on gate-model quantum computers}
\label{app-sec:qnd}
We detail the second (more efficient) implementation scheme described in Sec.~\ref{sec:realization} and show how the QR's convergence property leads to more efficient versions of both schemes in Sec.~\ref{sec:realization}.

If qubit reset is available, we can implement the second scheme in Sec.~\ref{sec:realization} based on quantum non-demolition (QND) measurements \cite{braginsky1995quantum}. In this scheme, to estimate $\langle Z^{(i)} \rangle_l$,  we no longer need to re-initialize and re-run the $N_m$ circuits from time $1$. Instead, for each $N_m$ circuits we perform a QND measurement of $Z^{(i)}$ at time $l$, continue running the circuit until the next QND measurement, and so forth. QND measurements ensure the information encoded in $\rho_{l}$ is retained from one timestep to the next. To process a length-$L$ input sequence, each $N_m$ circuits is run $S$ shots so that the average of $N_m S$ measurements at time $l$ estimates $\langle Z^{(i)} \rangle_l$. That is, this scheme needs $N_m S L$ applications of $T(u_l)$, but only $N_m S$ circuit runs compared to $N_m L S$ runs in the first scheme (see Sec.~\ref{sec:realization}). This presents a substantial saving as the number of circuit runs is independent of the input sequence length $L$.

To explain QND measurements, we first show that direct measurement of $Z$ on a ``system'' qubit is equivalent to coupling the qubit with an ancilla qubit via CNOT and measuring $Z_{\rm a}$, the Pauli $Z$ operator acting on the ancilla qubit ``a'' \cite{NC10}. To see this, let $| \psi \rangle_{\rm sys} = \alpha |0\rangle_{\rm sys} + \beta | 1 \rangle_{\rm sys}$ be the state of the system qubit. Prepare the ancilla qubit at the ground state $|0\rangle_{\rm a}$. We write
\begin{equation*}
{\rm CNOT} = | 0 \rangle \langle 0|_{\rm sys} \otimes I_{\rm a} + | 1 \rangle \langle 1|_{\rm sys} \otimes X_{\rm a}, 
\end{equation*}
where $I_{\rm a}$ and $X_{\rm a}$ are the identity and Pauli $X$ operators acting on the ancilla qubit. The system and ancilla state after applying CNOT is
\begin{equation*}
| \Psi \rangle = {\rm CNOT} | \psi \rangle_{\rm sys} \otimes | 0 \rangle_{\rm a} = \alpha |00\rangle + \beta |11\rangle.
\end{equation*}

\begin{center}
\begin{figure}[!ht]
\begin{equation*}
\hspace{2em} \Qcircuit @C=0.8em @R=.5em {
    \lstick{| \psi \rangle_{l-1}}       &\gate{U'_l}  &\multigate{1}{\mathcal{C}} &\qw    &\qw                   &\gate{U'_{l+1}}  &\multigate{1}{\mathcal{C}}   &\qw \\
    \lstick{| 0 \rangle^{\otimes n}} &\qw          &\ghost{\mathcal{C}}        &\meter &\push{| 0 \rangle^{\otimes n}}   &\qw              &\ghost{\mathcal{C}}          &\meter 
}
\end{equation*}
\caption{Quantum circuit implementing the QND measurements by coupling ancilla qubits $|0\rangle^{\otimes n}$ with the QR system qubits $| \psi \rangle_{l-1}$.}
\label{fig:QND}
\end{figure}
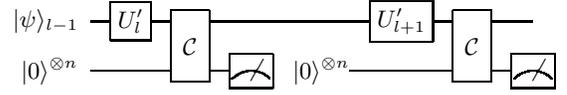
\end{center}

Measurement of $Z_{\rm a}$ on the ancilla qubit is described by the projectors $P_{+} = I_{\rm sys} \otimes | 0 \rangle \langle 0 |_{\rm a}$ and $P_{-} =  I_{\rm sys} \otimes | 1 \rangle \langle 1 |_{\rm a}$. Therefore, the probabilities and post-measurement system states are
\begin{equation*}
\begin{split}
& {\rm Pr}(+) = \langle \Psi | P_{+} | \Psi \rangle = |\alpha|^2,  \ \frac{{\rm Tr}_{\rm a} \left( P_{+} | \Psi \rangle \langle \Psi | P_{+} \right)}{{\rm Pr}(+)} = |0 \rangle \langle 0|_{\rm sys}, \\
& {\rm Pr}(-) = \langle \Psi | P_{-} | \Psi \rangle = |\beta|^2,  \ \frac{{\rm Tr}_{\rm a} \left( P_{-} | \Psi \rangle \langle \Psi | P_{-} \right)}{{\rm Pr}(-)} = |1 \rangle \langle 1|_{\rm sys}, \\
\end{split}
\end{equation*}
where ${\rm Tr}_{\rm a}(\cdot)$ is the partial trace over the ancilla qubit.

Now for an $n$-qubit QR, we associate each system qubit in the QR with its ancilla qubit. All $n$ ancilla qubits are prepared in the ground state. Suppose that when restricted to pure state preparation, we have drawn $N_m$ circuits using Monte Carlo sampling. For each of the $N_m$ circuits and each time step $l$, we apply the aforementioned ancilla-coupled measurement of $Z^{(i)}$ for each system qubit in the QR. After measuring the $n$ ancilla qubits, we reset and re-prepare them in the ground state for measurements at next time $l+1$; see Fig.~\ref{fig:QND}.

In Fig.~\ref{fig:QND}, $| \psi \rangle_{l-1}$ denotes the state of the system (QR) qubits and $|0 \rangle^{\otimes n}$ denotes the ancilla qubits initialized in the groud state. Here we have grouped the system and ancilla qubits and represent them using single wires. The unitary operator $U'_l$ is $U_0$ or $U_1$ with probabilities $(1-\epsilon)u_l$ and $(1-\epsilon)(1-u_l)$ and $U(l) = U'_l \mathcal{C}$, where $\mathcal{C}$ is a product of $n$ CNOT gates each acting on the $i$-th system-ancilla qubit pair. Measuring $Z^{(i)}_{\rm a}$ on the $i$-th ancilla qubit and resetting it at each time step $l = 1, \ldots, L$ is equivalent to having $L$ ancilla qubits associated to the $i$-th system qubit and measuring $Z^{(i)}_{{\rm a}, l}$ (i.e., the $l$-th ancilla qubit associated to the $i$-th system qubit). The resulting QR dynamics is
\begin{equation*}
T(u_l) \rho_{l-1} = (1-\epsilon)\left( u_l T_0 + (1-u_l) T_1 \right) \rho_{l-1} + \epsilon \sigma,
\end{equation*}
where $T_j (\rho_{l-1}) = {\rm Tr}_{\rm A}(U_j \mathcal{C} \rho_{l-1} \otimes (| 0 \rangle \langle 0|)^{\otimes n} \mathcal{C}^\dagger U_j^\dagger)$ for $j=0, 1$, and ${\rm Tr}_{\rm A}(\cdot)$ is the partial trace over all $n$ ancilla qubits denoted by ``A''.

We now show that the measured observables $Z^{(i)}_{a, l}$ commute at different times as required by QND. More generally, we will show that $Z_{{\rm a}, l} = \bigotimes_{i=1}^{n} O^{(i)}_{{\rm a}, l}$ ($l=1,\ldots,L$), where for each $i$ we have $O^{(i)}_{{\rm a}, l} = I^{(i)}$ (the identity operator on the $i$-th qubit) or $O^{(i)}_{{\rm a}, l} = Z^{(i)}_{{\rm a}, l}$, are QND observables. Firstly, we have the commutator $[Z_{a, k}, Z_{a, j}] = 0$ for all $k, j = 1, \ldots, L$. Denote the evolved observables in the Heisenberg picture by
\begin{equation*}
Z_{\rm a}(l) = U(1)^\dagger \cdots U(l)^\dagger Z_{{\rm a}, l} U(l) \cdots U(1)  = U_{l:1}^\dagger Z_{{\rm a}, l} U_{l:1},
\end{equation*}
where $U_{l:1} = U(l) \cdots U(1)$. For $k,j = 1, \ldots, L$ with $j < k$, we have
\begin{equation*}
\begin{split}
& [Z_{\rm a}(j), Z_{\rm a}(k)] \\
& = U_{j:1}^\dagger Z_{{\rm a}, j} U_{j:1} U_{k:1}^\dagger Z_{{\rm a}, k} U_{k:1} -  U_{k:1}^\dagger Z_{{\rm a}, k} U_{k:1} U_{j:1}^\dagger Z_{{\rm a}, j} U_{j:1} \\
& = U_{j:1}^\dagger  Z_{{\rm a}, j} U_{k:j+1}^\dagger  Z_{{\rm a}, k} U_{k:1} -  U_{k:1}^\dagger  Z_{{\rm a}, k} U_{k:j+1}  Z_{{\rm a}, j} U_{j:1} \\
& = U_{k:1}^\dagger [Z_{{\rm a}, j}, Z_{{\rm a}, k}] U_{k:1} = 0,
\end{split}
\end{equation*}
where in the second last equality we have used the fact that $ Z_{{\rm a}, j}$ commutes with the future unitary operations $U_{k:j+1}$. If $j > k$, apply the same argument as above shows $[Z_{\rm a}(j), Z_{\rm a}(k)] = - [Z_{\rm a}(k), Z_{\rm a}(j)] = 0$. The commutativity of $Z_{\rm a}(j)$ and $Z_{\rm a}(k)$ for all $j, k \geq 1$ means that the sequence $\{ Z_{\rm a}(j), j=1,2,\ldots\}$ has a joint probability distribution and constitutes a classical stochastic process. QND measurements on the sequence gives a realization of this stochastic process. 

The QR's convergence property (see Appendix~\ref{app-subsec:cv}) leads to more efficient versions of both schemes in Sec.~\ref{sec:realization}. Let $M$, $\rho_l$, $\rho_{l-M}$ and $\tilde{\rho}_l$ be as given in Sec.~\ref{sec:realization}. By the convergence property (Eq.~\ref{eq:cv}), we have 
\begin{equation}\label{app-eq:cv-realization}
\| \rho_l - \tilde{\rho}_l \|_1 \leq (1-\epsilon)^{M} \| \rho_{l-M} - (| 0 \rangle \langle 0 |)^{\otimes n}\|_1 \leq 2(1-\epsilon)^{M}.
\end{equation}
The difference between $\rho_l$ and $\tilde{\rho}_l$ can be made negligible by choosing $M$ appropriately based on $\epsilon$. If we perform repeated measurements on $\rho_l$ and $\tilde{\rho}_l$, then the finite-sample estimates of $\langle \widetilde{ Z^{(i)} }\rangle_l = {\rm Tr}(\tilde{\rho}_l Z^{(i)})$ and $\langle Z^{(i)} \rangle_l = {\rm Tr}(\rho_l Z^{(i)})$ will also be close; see Appendix~\ref{app-sec:mc-estimate}.
Using the convergence property, the first scheme in Sec.~\ref{sec:realization} requires $N_m S L$ circuit runs but only $N_m S L M$ applications of $T(u_l)$. When $L>M$, a substantial saving in the number of applications of $T$ can be obtained (for timesteps $l>M$) compared to the previous quadratic dependence on $L$. The second scheme now only requires at most $N_m S M$ applications of $T(u_l)$ and $N_m S$ circuit runs, both are \textit{independent} of the input sequence length $L$. This provides a path for fast and large scale temporal processing using QRs.

\section{Monte Carlo estimation}
\label{app-sec:mc-estimate}
For all schemes described in Sec.~\ref{sec:realization}, we can set $S=1$ and run $N_m$ Monte Carlo sampled circuits (possibly in parallel if many  copies of the same hardware are available) for a sufficiently large $N_m$. We show that the average of all $N_m$ measurements at time $l$ estimates $\langle Z^{(i)} \rangle_l$ and its variance vanishes as $N_m$ tends to infinity.

First consider estimating $\langle Z^{(i)} \rangle_l$ by re-initializing each $N_m$ circuit in $|0\rangle^{\otimes n}$ and re-running them from time 1 to time $l$ according to inputs  $\{u_{1}, \ldots, u_{l}\}$}. Recall that
$$\langle Z^{(i)} \rangle_l = {\rm Tr}( Z^{(i)}  \rho_l) = {\rm Tr}( Z^{(i)} T(u_l) \cdots T(u_1) (|0\rangle \langle 0|)^{\otimes n}),$$ where $T(u_k)$ is the input-dependent CPTP map defined in Eq.~\eqref{eq:dynamics} for $k = 1, \ldots, l$. Define independent discrete-valued random variables $X_{k}$ such that 
\begin{equation*}
\begin{split}
& {\rm Pr}(X_{k} = 0) = (1-\epsilon) u_{k},\\
& {\rm Pr}(X_{k} = 1) = (1-\epsilon)(1-u_{k}), \\
& {\rm Pr}(X_{k} = 2) = \epsilon. \\
\end{split}
\end{equation*}
To implement the QR, for each time $k$, we independently sample $N_m$ random variables $X_{k, j}$ $(j=1,\ldots,N_m)$ from the same distribution as $X_k$. 
Define
\begin{equation*}
T_{x} =\left\{  \begin{array}{cc} T_0, &  \hbox{if $x= 0$},\\   
  T_1,  & \hbox{if $x= 1$},\\
 K_{\sigma}, & \hbox{if $x= 2$}, \end{array} \right.
\end{equation*}
where $K_{\sigma}(\rho)  = \sigma$ is a constant CPTP map that sends any density operator $\rho$ to the constant density operator $\sigma$ in Eq.~\eqref{eq:dynamics-subsystem}. The random CPTP maps $T_{X_{k, j}}$ follow the same distribution as $X_{k, j}$ and are independent for each $k$ and $j$. Furthermore, $\mathbb{E}[T_{X_{k, j}}]= T(u_k)$.

For the $j$-th circuit, we implement a sequence of (random) CPTP maps $T_{X_{l, j}} \cdots T_{X_{1, j}}$ so that at time $l$, the (random) QR state is 
$$\rho^{\boldsymbol{X}_{l, j}} = T_{X_{l, j}} \cdots T_{X_{1, j}} (| 0 \rangle \langle 0|)^{\otimes n},$$
where $\boldsymbol{X}_{l,j} = (X_{1,j}, \ldots, X_{l,j})$.  For each $j$-th circuit, we measure $Z^{(i)}$ and denote its random outcome by $\overline{ Z^{(i)} }_{l,j}$. Note that for $j=1,\ldots, N_m$, $\overline{ Z^{(i)} }_{l,j}$ are independent (but not necessarily identically distributed) random variables taking values $\pm 1$ (eigenvalues of $Z^{(i)}$) with conditional probabilities (conditional on the random variables $\boldsymbol{X}_{l, j}$)
\begin{equation*}
{\rm Pr}\left(\overline{ Z^{(i)}}_{l,j} = z |  \boldsymbol{X}_{l, j} \right) = {\rm Tr}\left( \rho^{\boldsymbol{X}_{l, j}}  P^{(i)}_{z} \right),  \quad z = \pm 1,
\end{equation*}
where $P^{(i)}_{\pm 1}$ are the projectors such that $Z^{(i)} = P^{(i)}_{+ 1} - P^{(i)}_{-1}$. Consider the average of all $N_m$ measurement outcomes, by the law of total expectation,
\begin{equation*}
\begin{split}
& \frac{1}{N_m} \sum_{j=1}^{N_m} \mathbb{E} \left[ \overline{Z^{(i)}}_{l,j} \right] \\
& = \frac{1}{N_m} \sum_{j=1}^{N_m} \mathbb{E} \left[ \mathbb{E}\left[ \overline{ Z^{(i)} }_{l,j}  | \boldsymbol{X}_{l, j} \right] \right] \\
& = \frac{1}{N_m} \sum_{j=1}^{N_m} \mathbb{E}\left[ {\rm Tr}\left( Z^{(i)}  \rho^{\boldsymbol{X}_{l, j}} \right) \right] \\
& = \frac{1}{N_m} \sum_{j=1}^{N_m} {\rm Tr}\left( Z^{(i)}  \mathbb{E}[T_{X_{l, j}}] \cdots \mathbb{E}[T_{X_{1, j}}] (| 0 \rangle \langle 0 |)^{\otimes n} \right)  \\
& = \frac{1}{N_m} \sum_{j=1}^{N_m} {\rm Tr}\left( Z^{(i)}  T(u_l) \cdots T(u_1) (| 0 \rangle \langle 0 |)^{\otimes n} \right) \\
& = {\rm Tr}(\rho_{l} Z^{(i)}) = \langle Z^{(i)} \rangle_l,
\end{split}
\end{equation*}
therefore the finite-sample estimate is unbiased. Moreover, using the fact that 
$$\mathbb{E}\left[\left( \overline{ Z^{(i)} }_{l,j}\right)^2\right] = \sum_{z=\pm1} z^2 {\rm Pr}\left(\overline{ Z^{(i)}}_{l,j} = z \right) =1,$$ 
the variance of the average of $N_m$ measurements is
\begin{equation*}
\begin{split}
& {\rm Var}\left[\frac{1}{N_m} \sum_{j=1}^{N_m} \overline{  Z^{(i)} }_{l,j} \right]\\
& = \frac{1}{N^2_m} \sum_{j=1}^{N_m} {\rm Var}\left[ \overline{ Z^{(i)} }_{l,j} \right] = \frac{1}{N_m} \left(1 - \langle Z^{(i)}\rangle^2_l\right).
\end{split}
\end{equation*}

Using the convergence property, to estimate $\langle Z^{(i)} \rangle_l$ for a sufficiently large $l$ (that depends on $\epsilon$), we re-initialize $N_m$ circuits at time $l-M$ and run the circuits according to inputs $\{u_{l-M+1}, \ldots, u_l \}$. Let $\langle \widetilde{Z^{(i)}} \rangle_l = {\rm Tr}(Z^{(i)} \tilde{\rho}_l)$ where
$$\tilde{\rho}_l = T(u_l) \cdots T(u_{l-M+1}) (|0 \rangle \langle 0|)^{\otimes n}.$$ 
In this setting, for the $j$-th circuit, we implement a sequence of (random) CPTP maps $T_{X_{l,j}} \cdots T_{X_{l-M+1, j}}$ so that the (random) QR state at time $l$ is
$$
\rho^{\widetilde{\boldsymbol{X}}_{l, j}} = T_{X_{l, j}} \cdots T_{X_{l-M+1, j}} (|0 \rangle \langle 0 |)^{\otimes n},
$$
where $\widetilde{\boldsymbol{X}}_{l, j} = (X_{l-M+1, j}, \ldots, X_{l, j})$. Let $\widetilde{Z^{(i)}}_{l, j}$ be the random outcome of measuring $Z^{(i)}$. The conditional probabilities are
$$
{\rm Pr}\left( \widetilde{Z^{(i)}} = z | \widetilde{\boldsymbol{X}}_{l, j} \right) = {\rm Tr}\left(\rho^{\widetilde{\boldsymbol{X}}_{l, j}} P^{(i)}_z \right), \quad z = \pm 1.
$$
A similar argument as above shows that the average of all $N_m$ measurements satisfies
\begin{equation*}
\begin{split}
\mathbb{E}\left[ \frac{1}{N_m} \sum_{j=1}^{N_m} \widetilde{Z^{(i)}}_{l, j} \right] & = \langle \widetilde{Z^{(i)}} \rangle_{l}, \\
{\rm Var}\left[ \frac{1}{N_m} \sum_{j=1}^{N_m} \widetilde{Z^{(i)}}_{l, j} \right] & = \frac{1}{N_m} \left(1 -  \langle \widetilde{Z^{(i)}} \rangle_{l}^2 \right).
\end{split}
\end{equation*}
The convergence property and Eq.~\ref{app-eq:cv-realization} ensure that the bias (in mean) vanishes exponentially fast,
\begin{equation*}
\begin{split}
& \left| \mathbb{E}\left[ \frac{1}{N_m} \sum_{j=1}^{N_m} \widetilde{Z^{(i)}}_{l, j} \right] - \langle Z^{(i)} \rangle_l \right| \\
& = \left| {\rm Tr} \left( Z^{(i)} (\tilde{\rho}_{l} - \rho_l) \right) \right| \\
& \leq \| \tilde{\rho}_l - \rho_l \|_1 \leq 2(1-\epsilon)^M,
\end{split}
\end{equation*}
where we have used the fact that for any Hermitian matrix $A$, $| {\rm Tr}(Z^{(i)} A)| \leq \sigma_{\max}(Z^{(i)}) \|A\|_1$, with $\sigma_{\max}(Z^{(i)})=1$ denotes the maximum singular value of $Z^{(i)}$. This shows that the bias can be exponentially suppressed by choosing $M$ appropriately based on $\epsilon$, so that the estimates of $\langle \widetilde{Z^{(i)}} \rangle_l$ and $\langle Z^{(i)} \rangle_l$ are also close.

\section{Experimental and numerical details}
\label{app-sec:exp}
\subsection{Nonlinear temporal processing tasks}
\label{app-subsec:tasks}
We give detailed descriptions of the five nonlinear temporal processing tasks. Tasks~\Romannum{1} and \Romannum{2} are governed by linear reservoirs with polynomial readout \cite{BC85,grigoryeva2018universal}, described by
\begin{equation*}
\begin{cases}
\textbf{x}_{l} = A \textbf{x}_{l-1} + c u_l \\
y_{l} = h(\textbf{x}_l),
\end{cases}
\end{equation*}
where $A \in \mathbb{R}^{2000 \times 2000}$ and $c \in \mathbb{R}^{2000}$. To have short-term or fading memory, we rescale the maximum singular value $\sigma_{\max}(A) = 0.5$ for Task~\Romannum{1} and $\sigma_{\max}(A) = 0.99$ for Task~\Romannum{2}, meaning that Task~\Romannum{2} retains the initial condition and past inputs for a longer time duration. The sparsity of $A$ determines the pairwise correlation between reservoir state elements. We set $A$ to be a full (dense) matrix for Task~\Romannum{1} and 95\% sparse for Task~\Romannum{2}. The readout function $h$ is a degree $2$ polynomial in the state elements. Task~\Romannum{3} is a recently proposed classical reservoir computing model that achieves good performance in chaotic system modeling \cite{grigoryeva2018universal}, described by
\begin{equation*}
\begin{cases}
\textbf{x}_{l} = p(u_l) \textbf{x}_{l-1} + q(u_l) \\
y_l = \textbf{w}^T \textbf{x}_l,
\end{cases}
\end{equation*}
where $p(u_l) = \sum_{j=0}^{4} A_j u^j_l$ and $q(u_l) = \sum_{j=0}^{2} B_j u^j_l$ are matrix-valued polynomials in the input $u_l$, $A_j \in \mathbb{R}^{700 \times 700} \bigoplus \mathbb{R}^{700 \times 700}$ and $B_j \in \mathbb{R}^{700 \times 1} \bigoplus \mathbb{R}^{700 \times 1}$. For Task~\Romannum{3}, We rescale $\sigma_{\rm max}(A_j) < \frac{1}{3}$ for all $j$ so that both it exhibits short-term memory. Task~\Romannum{4} is a Volterra series with kernel order 5 and memory 2, commonly applied to model responses of nonlinear systems in control engineering \cite{BC85},
\begin{equation*}
y_l = \textbf{w}_c + \sum_{i=1}^{5} \sum_{j_1, \ldots, j_i=0}^{2} \textbf{w}^{j_1, \ldots, j_i}_i \prod_{k=1}^{i} u_{l-j{_k}}.
\end{equation*}
For the first three tasks, elements of $A, A_j, B$ and $\textbf{w}$ are uniformly randomly sampled from $[-1, 1]$. The constant $c$ and coefficients of readout function $h$ are also sampled from the same distribution. The same applies to the kernel coefficients $\textbf{w}^{j_1, \ldots, j_i}_i$ and $\textbf{w}_c$ in Task~\Romannum{4}.

Task~\Romannum{5} is a missile moving with a constant velocity in the horizontal plane, a continuous-time long-term memory nonlinear map \cite{ni1996new} described by
\begin{equation*}
\begin{cases}
\dot{\textbf{x}}_1 = \textbf{x}_2 - 0.1 \cos(\textbf{x}_1)(5 \textbf{x}_1 - 4 \textbf{x}^3_1 + \textbf{x}^5_1) - 0.5 \cos(\textbf{x}_1) u \\
\dot{\textbf{x}}_2 = -65 \textbf{x}_1 + 50 \textbf{x}_1^3 - 15 \textbf{x}^5_1 - \textbf{x}_2 - 100 u,
\end{cases} 
\end{equation*}
with $y= \textbf{x}_2$. This missile dynamics is simulated using the $(4, 5)$ Runge-Kutta formula in MATLAB, with a sampling time of $\tau = 1/80$ for 1 second. 

\subsection{Quantum circuits for QRs}
\label{app-subsec:circuits}
We detail the circuits implementing the QR dynamics in our proof-of-principle experiments presented in Sec.~\ref{sec:exp}. The quantum circuits for the 4-qubit and 10-qubit Boeblingen QRs are shown in Fig.~\ref{fig:7}. The quantum circuits for the 5-qubit Ourense and 5-qubit Vigo QRs are shown in Fig.~\ref{fig:8}.

\begin{figure*}[!ht]
\includegraphics[scale=0.9]{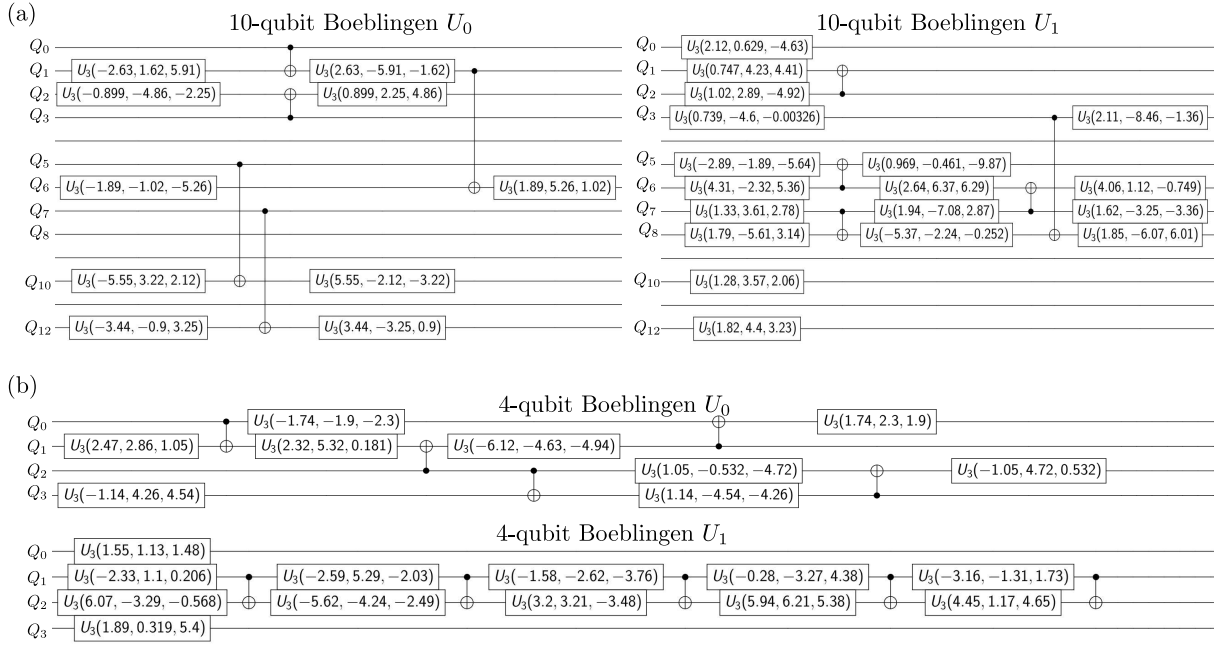}
\caption{Quantum circuits for the (a) 10-qubit QR and (b) 4-qubit QR on the Boeblingen device.}
\label{fig:7}
\end{figure*}

\begin{figure*}[!ht]
\includegraphics[scale=0.9]{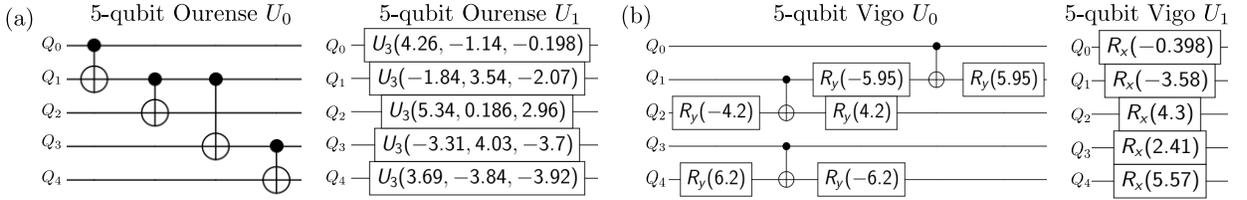}
\caption{Quantum circuits for the (a) 5-qubit Ourense QR and (b) 5-qubit Vigo QR.}
\label{fig:8}
\end{figure*}

\subsection{Full input-output sequential data}
\label{app-subsec:data}
Since we bypass the washout for QRs by initializing them in the state $|0 \rangle^n$, this is equivalent to washing out their initial conditions with a length $L_w$ constant input sequence $u_l = 1$. The same washout has been applied to all nonlinear tasks. We have checked that $L_w=50$ is enough for all tasks to reach steady states given the same initialization $\textbf{x}_0 = 0$. Particular caution has been taken to washout Task~\Romannum{4}, in which we set $u_l = 1$ for $l=-2, -1$. For each target map, we discard the first four input-output sequence data points, and the corresponding QR experimental data, to remove the transitory output response due to the change in input statistics. In Fig.~\ref{fig:9}, we show the full washout, train and test input-output target sequences for both the multi-step ahead prediction and the map emulation problems. Fig.~\ref{fig:10} plots the full target output sequences, the train and test QR outputs on the multi-step ahead prediction problem. Fig.~\ref{fig:11} plots the full target output sequences, the train and test QR outputs on the map emulation problem. In all figures, the transitory responses are indicated by dotted lines.
\begin{figure*}[!ht]
\centering
\includegraphics[scale=0.8]{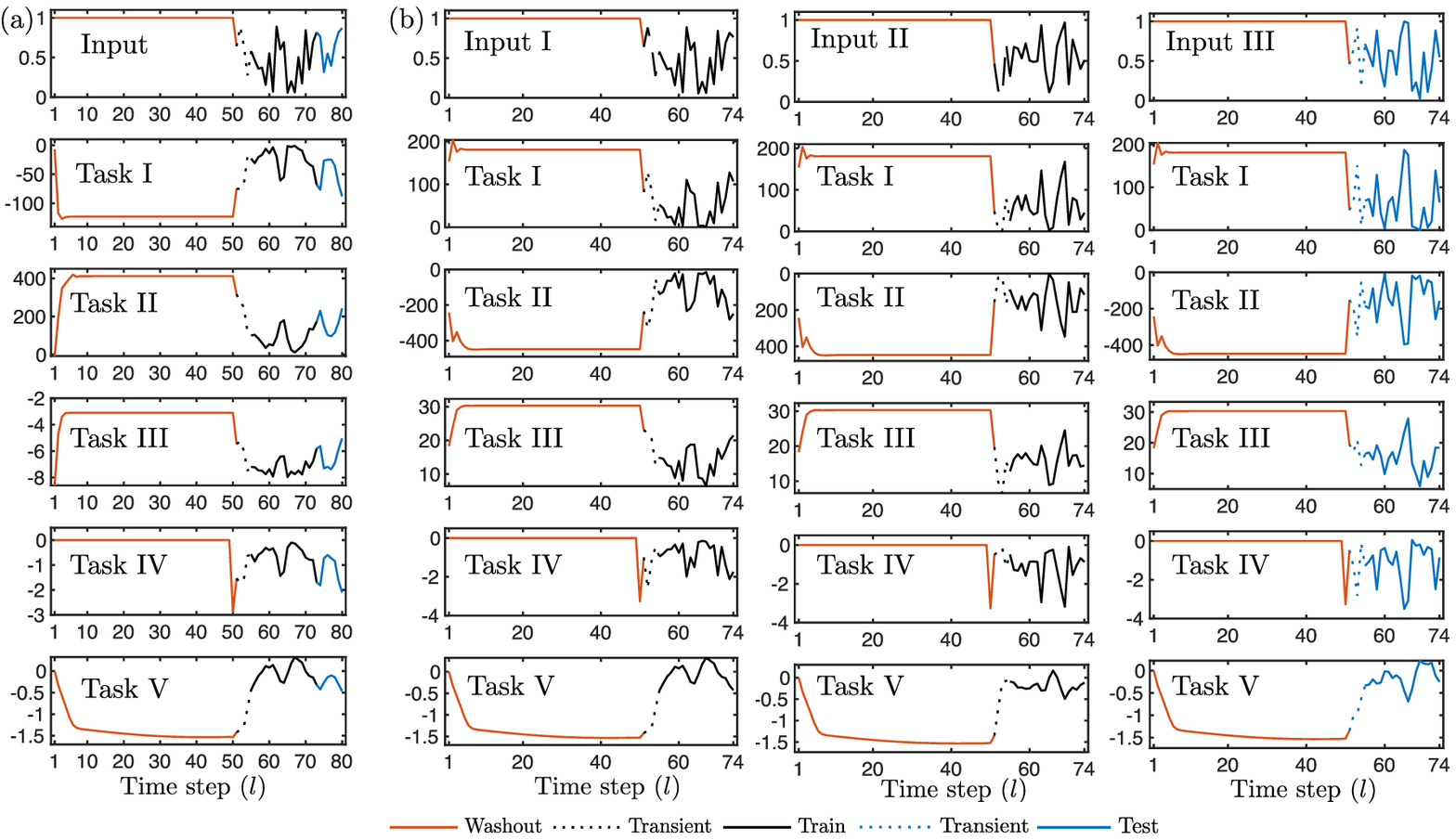}
\caption{Full washout, train and test input-output sequences for (a) The multi-step ahead prediction problem and (b) The map emulation problem. The first row in (a) and (b) shows the input sequences.}
\label{fig:9}
\end{figure*}

\begin{figure*}[!ht]
\centering
\includegraphics[scale=0.8]{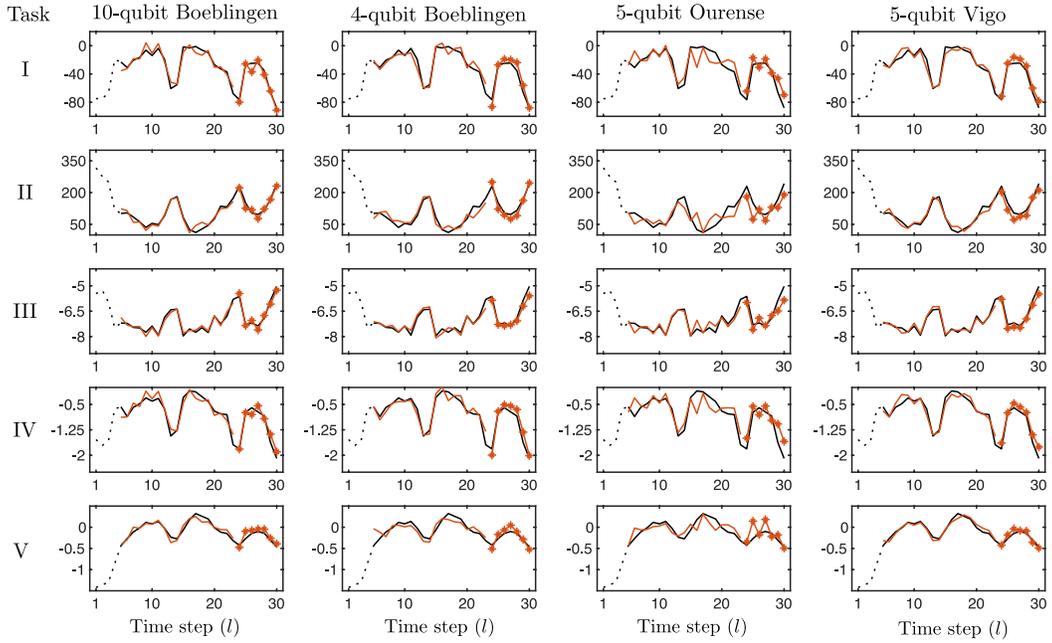}
\caption{The full target output sequences, the train and test output sequences of the four QRs for each task on the multi-step ahead prediction problem. Each column corresponds to each $n$-qubit QR outputs. Each row corresponds to each task.}
\label{fig:10}
\end{figure*}

\begin{figure*}[!ht]
\centering
\includegraphics[scale=0.9]{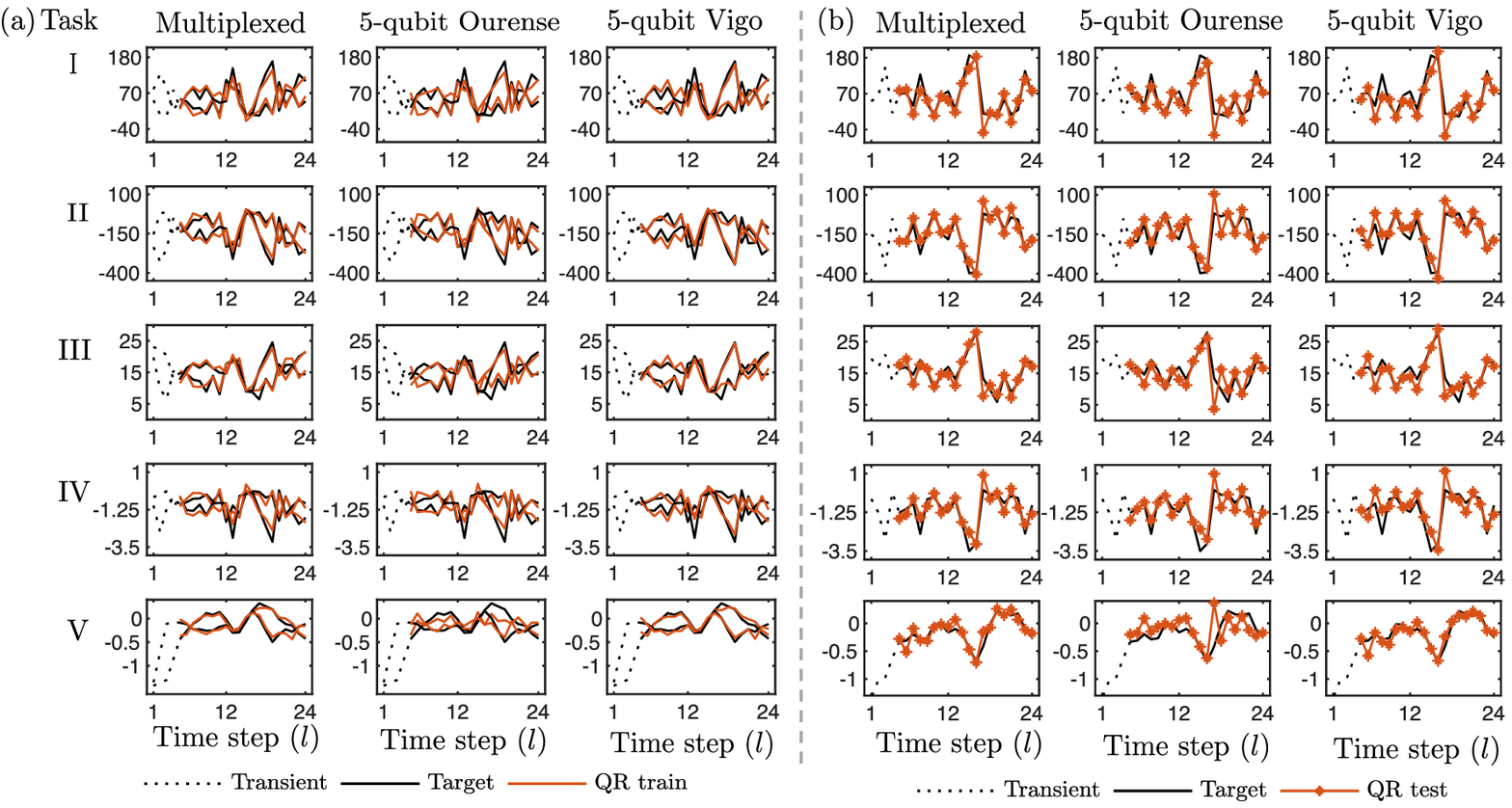}
\caption{The full target output sequences, the train and test output sequences of the QRs for each task on the map emulation problem. (a) Shows the two train output sequences. (b) Shows the test output sequence. The columns correspond to the multiplexed 5-qubit QRs, 5-qubit Ourense QR and the 5-qubit Vigo QR from the left to the right. Each row corresponds to each task.}
\label{fig:11}
\end{figure*}

\subsection{Measurement and simulation data}
\label{app-subsec:measurement-data}
We simulate the four QRs using the IBM Qiskit simulator under ideal and noisy conditions. The noise models used are obtained from the device calibration data. We fetched the updated device calibration data each time a job was executed on the hardware. The circuits simulated are the same as the circuits employed for the experiments and so is the number of shots. For the multi-step ahead prediction problem, the 10-qubit Boeblingen QR experienced a significant deviation from simulated results on qubits $Q=1, 8$ (see Fig.~\ref{fig:12}), resulting in larger NMSE $=0.26, 0.29, 0.068, 0.15, 6.1$ for the four tasks. After setting the readout parameters $\textbf{w}_1 = \textbf{w}_8=0$ for $Q=1, 8$, this issue was circumvented at the cost of using a fewer number of computational features. The resulting 10-qubit Boeblingen QR still achieves performance improvement over other QRs with a smaller number of qubits on the multi-step ahead prediction problem in the first three tasks. A time-invariant readout error in qubit $i$ linearly transforms the expectation $\langle Z^{(i)} \rangle_{l}$. The QR predicted outputs are invariant under time-invariant readout errors when using linear regression to optimize $\textbf{w}, \textbf{w}_c$ as derived in Appendix~\ref{app-sec:readout-error}. However, for the 10-qubit Boeblingen QR, the deviations in qubits $Q=1, 8$ were time-varying. On the other hand, the 5-qubit Vigo device experienced almost time-invariant deviations in qubit $Q=0$ as shown in Figs.~\ref{fig:12} and \ref{fig:13}, but this does not affect the performance of this QR noticeably. The experimental results of the 5-qubit Ourense QR follow the noisy simulation results closely. For the map emulation problem, the experimental results of both 5-qubit QRs follow the simulated results closely, with an almost time-invariant shift in $Q=0$ for the 5-qubit Vigo QR.

\begin{figure*}[!ht]
\includegraphics[scale=0.9]{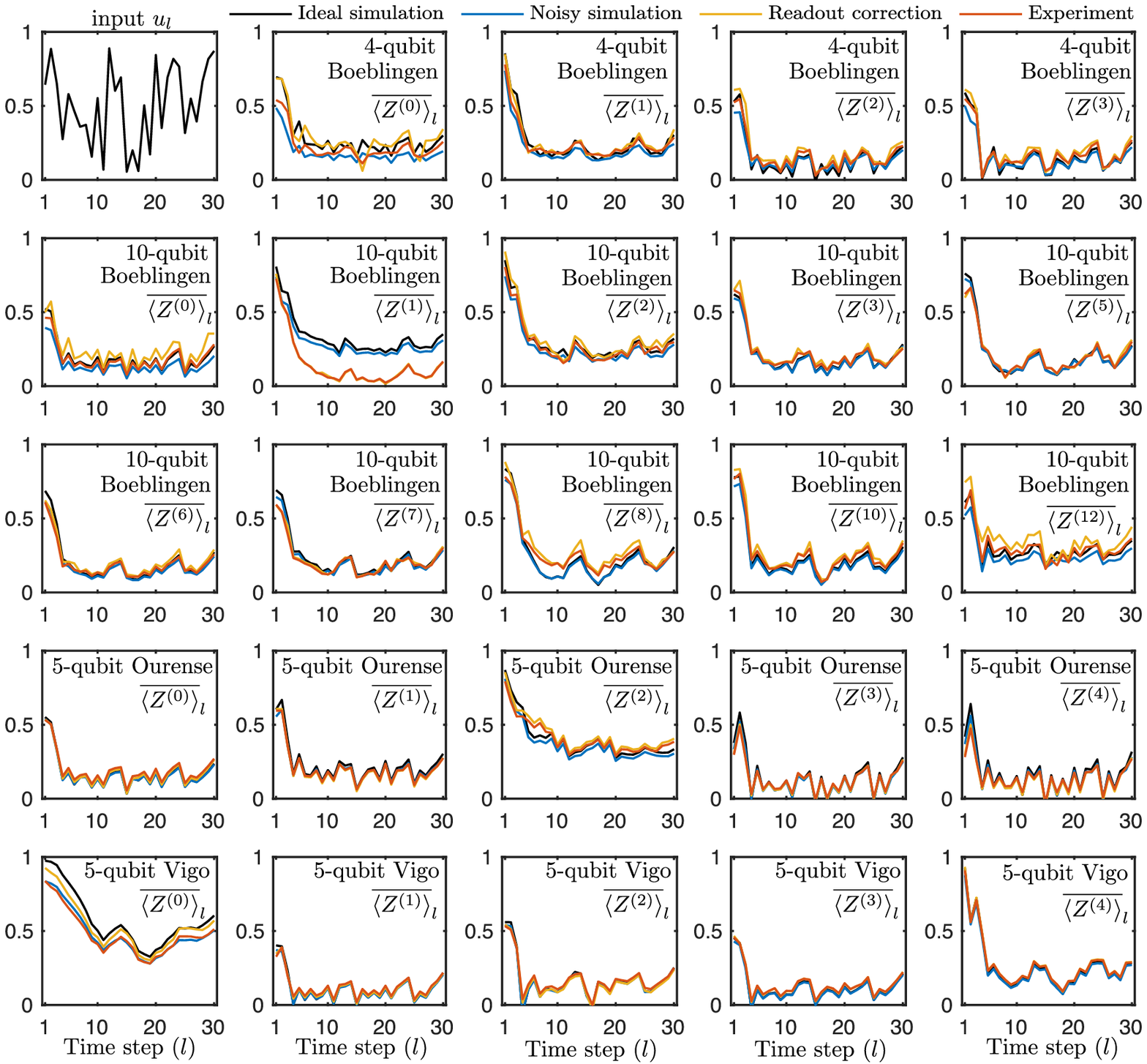}
\caption{Input sequence, experimental and simulation results for each qubit of the four QRs at each time step $l=1, \ldots, 30$, for the multi-step ahead prediction problem.}
\label{fig:12}
\end{figure*}

\begin{figure*}[!ht]
\includegraphics[scale=0.9]{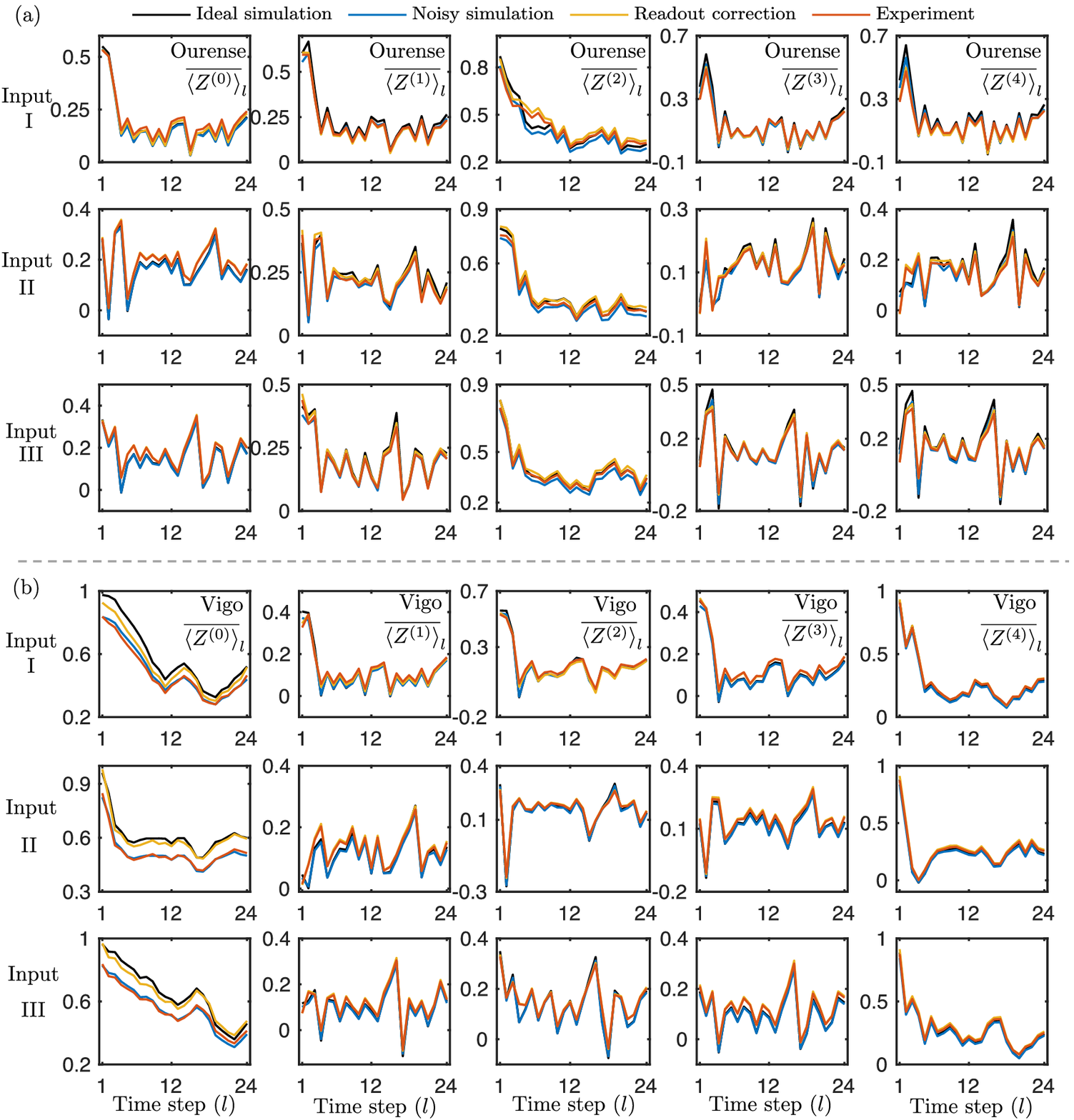}
\caption{Experimental and simulation results for each qubit $i=0, \ldots, 4$ and each time step $l=1, \ldots, 24$, for the map emulation problem. Three input sequences are used in this problem, labeled as inputs I, II and III. Row $i$ in a sub-figure corresponds to the experimental data for the $i$-th input sequence. Column $j$ corresponds to the experimental data for the $j$-th qubit. (a) Shows the experimental data for the 5-qubit Ourense QR. (b) Shows the experimental data for the 5-qubit Vigo QR.}
\label{fig:13}
\end{figure*}

\subsection{Hardware specifications}
\label{app-subsec:hardware}
The experiments were conducted on the IBM 20-qubit Boeblingen (version 1.0.0), 5-qubit Ourense (version 1.0.0) and 5-qubit Vigo (version 1.0.0) superconducting quantum processors \cite{IBMQ}. The gate duration for an arbitrary single-qubit rotation gate $U_3$ \cite{cross2017open} is $\tau_{U_3} \approx 71.1$ ns for all qubits whereas the CNOT gate durations differ for different qubits. 

See Fig.~\ref{fig:7} for the 4-qubit and 10-qubit Boeblingen QR quantum circuits. The circuits are chosen such that both QRs have the same number of layers in $U_0$ and $U_1$. In this setting, the maximum duration of a circuit executed on the Boeblingen device is the same for both QRs. As stated in the main article, the chosen qubits for the 4-qubit QR and the 10-qubit QR on the Boeblingen device are $Q=0,1,2,3$ and $Q=0, 1, 2, 3, 5, 6, 7, 8, 10, 12$. These qubits were chosen due to their longer coherence times, shorter CNOT gate durations, smaller gate and readout errors. During the experiment, the maximum readout error was $10^{-2}$ and the maximum $U_3$ gate error implemented was $10^{-3}$. The maximum CNOT gate error implemented was $4.3\times10^{-2}$ and the maximum CNOT gate duration was $\tau_{\rm CNOT} \approx 427$ ns. We assume that commuting gates can be executed in parallel. We choose $N_0 = N_1 = 5$ numbers of layers for $U_0$ and $U_1$ in the 4-qubit and 10-qubit Boeblingen QRs. The maximum length of any input sequence (including the transient) for the multi-step ahead prediction and the map emulation problems is $L=30$. Therefore, the maximum numbers of $U_3$ gate executions and CNOT gate executions is $5L = 5 \times 30 = 150$. The maximum duration of a circuit executed on the Boeblingen device was $150 \times (\tau_{U_3} + \tau_{\rm CNOT}) \approx 150 \times (71.1 + 427) = 74.7 \ \mu$s, within the coherence times ($T_1, T_2$) for most qubits chosen.

Fig.~\ref{fig:8} shows the quantum circuits for the 5-qubit Ourense and 5-qubit Vigo QRs. Owing to the more restricted qubit couplings in these 5-qubit devices, the circuits for the  5-qubit QRs are simpler than that of the 4-qubit and 10-qubit Boeblingen QRs. To combine different computational features for the spatial multiplexing technique, we choose circuits that are sufficiently different for these two 5-qubit QRs. In particular, the 5-qubit Vigo QR consists of single-qubit rotational $Y$ gates in $U_0$ and single-qubit rotational $X$ gates in $U_1$. On the other hand, the 5-qubit Ourense QR uses arbitrary single-qubit rotational gates $U_3$ only in the circuit implementing $U_1$. 

The 5-qubit Ourense device achieves the same order of magnitude in readout errors, coherence times and CNOT gate durations as the 20-qubit Boeblingen device, but lower CNOT gate errors. For the Ourense device, the maximum $U_3$ gate error and readout error implemented were $0.9 \times 10^{-3}$ and $4.1 \times 10^{-2}$, and the maximum CNOT gate error implemented was $8 \times 10^{-3}$, a lower error compared to the Boeblingen device. The maximum CNOT gate duration implemented was $\tau_{\rm CNOT} \approx 576$ ns. For the 5-qubit Ourense QR, the circuit implementing $U_0$ is longer than that for $U_1$. The $U_0$ circuit consists of four CNOT gates, and the maximum duration of a circuit executed for the 5-qubit Ourense QR was $4L \times \tau_{\rm CNOT} \approx 70\,\mu$s, also within the coherence limits of most qubits.

The 5-qubit Vigo device is similar to the 5-qubit Ourense device. They have the same qubit couplings and share similar noise profile and hardware specifications. Rotational $X$ and $Y$ gates were used on this device, with gate duration $\tau = 35.5$ ns. The maximum single-qubit gate error implemented was $0.8 \times 10^{-3}$ and the maximum readout error implemented was $7.8 \times 10^{-2}$. The maximum CNOT gate error and gate duration implemented was $1.3 \times 10^{-2}$ and $\tau_{\rm CNOT} \approx 462.2$ ns, respectively. For this QR, $U_0$ is the longer circuit consisting of three layers of single-qubit rotation $Y$ gates and two layers of CNOT gates. Therefore, the maximum duration of a circuit implemented was $(3\tau+2\tau_{\rm CNOT})L =(3 \times 35.5 + 2 \times 462.2) \times 30 \approx 31 \ \mu$s.

\clearpage
\bibliography{prapplied.bib}{}

\end{document}